\documentclass[cleveref, autoref, thm-restate]{lipics-v2021}
\usepackage[sort]{cite}
\usepackage{tikz}
\usepackage{"mydefs"}
\usepackage{multirow}
\usepackage{apxproof}
\usepackage{colortbl}
\usepackage{thmtools} 
\usepackage{thm-restate}
\usepackage[justification=centering]{caption}
\usepackage{footnote}
\makesavenoteenv{tabular}
\usetikzlibrary{decorations.pathreplacing}
\nolinenumbers

\Copyright{Marco Carmosino, Ronald Fagin, Neil Immerman, Phokion Kolaitis, Jonathan Lenchner, Rik Sengupta, and Ryan Williams}


\authorrunning{Anonymized}
\authorrunning{M. Carmosino et al.}


\title{Parallel Play Saves Quantifiers}

\author{Marco Carmosino}{MIT-IBM Watson AI Lab}{mlc@ibm.com}{}{}
\author{Ronald Fagin}{IBM Almaden Research Center}{fagin@us.ibm.com}{}{}
\author{Neil Immerman}{University of Massachusetts Amherst}{immerman@umass.edu}{https://orcid.org/0000-0001-6609-5952}{}
\author{Phokion Kolaitis}{UC Santa Cruz and IBM Almaden Research Center}{kolaitis@ucsc.edu}{https://orcid.org/0000-0002-8407-8563}{}
\author{Jonathan Lenchner}{IBM T. J. Watson Research Lab}{lenchner@us.ibm.com}{https://orcid.org/0000-0002-9427-8470}{}
\author{Rik Sengupta}{MIT-IBM Watson AI Lab}{rik@ibm.com}{https://orcid.org/0000-0002-9238-5408}{}
\author{Ryan Williams}{Massachusetts Institute of Technology}{rrw@mit.edu}{}{}
\date{}

\begin{document}

\maketitle

\begin{abstract}
  The number of quantifiers needed to express first-order properties is captured by two-player combinatorial games called \emph{multi-structural} (MS) games. We play these games on linear orders and strings, and introduce a technique we call \emph{parallel play}, that dramatically reduces the number of quantifiers needed in many cases. Linear orders and strings are the most basic representatives of \emph{ordered} structures --- a class of structures that has historically been notoriously difficult to analyze. Yet, in this paper, we provide upper bounds on the number of quantifiers needed to characterize different-sized subsets of these structures, and prove that they are tight up to constant factors, including, in some cases, up to a factor of $1+\varepsilon$, for arbitrarily small $\varepsilon$.
\end{abstract}

\maketitle

\section{Introduction}\label{sec:intro}

How many quantifiers are needed to describe a property in first-order (FO) logic? While this is a very natural way to measure the complexity of expressing a FO property, until recently, in comparison to the quantifier \emph{rank} (a.k.a.~quantifier depth), there had been little work on this complexity measure. In 1981, Immerman~\cite{MScanon0} defined a two-player combinatorial game, which he called the \emph{separability game}, that completely characterizes the number of quantifiers needed to express a property. Immerman's work lay dormant until 2021, when Fagin, Lenchner, Regan and Vyas~\cite{MScanon1} rediscovered the games, called them \emph{multi-structural games} (which we shall call them), and introduced methods for analyzing them. Other recent follow-up works on these games include~\cite{MScanon2,MScanon3,vinallsmeeth2024quantifier}. Additional games to study number of quantifiers were recently introduced in \cite{hella2024}, and close cousins of multi-structural (MS) games were used to study formula size in \cite{DBLP:journals/corr/abs-1208-4803,DBLP:journals/lmcs/GroheS05}.

In this paper, we introduce a powerful way to analyze MS games, which we call \emph{parallel play}. We use parallel play to obtain non-trivial upper bounds on the number of quantifiers needed to separate different sets of linear orders and binary strings. We note that these are basic representatives of \emph{ordered} structures, a setting where combinatorial games have been historically notoriously difficult to analyze, a few exceptions (e.g.,\cite{DBLP:journals/apal/Schwentick96}) notwithstanding. In fact, methodological limitations of using combinatorial games on ordered structures have been explored in \cite{Chen2013OnLO}. In several cases, our upper bounds are within a $(1+\varepsilon)$ multiplicative factor of provable lower bounds, for arbitrarily small $\varepsilon > 0$. Since binary strings are universal encoding 
objects, such sets are extremely general and the results of this paper represent a notable advancement in a different direction than the initial results on linear orders, rooted trees, and $s$-$t$ connectivity from \cite{MScanon0, MScanon2}.

\medskip \noindent{\bf Multi-Structural Games. }
In order to replace quantifier rank with quantifier number, MS games modify the setup of classical {\ef} (EF) games. Firstly, MS games are played by the customary two perfect players, Spoiler ($\bS$) and Duplicator ($\bD$), on a pair of \emph{sets} of structures, rather than on a pair of structures. Secondly, MS games give $\bD$ more power by enabling her to make arbitrarily many copies of the structures on her side before responding, if she wishes. In order to win the game, she only needs to ensure that there are just two structures, one from each set, that ``match'' each other at the end of the game.

The fundamental theorem for MS games~\cite{MScanon0, MScanon1} (see Theorem~\ref{thm:MSfundamental}) states that $\bS$ wins the $r$-round MS game on a pair $(\cA, \cB)$ of sets of structures if and only if there is an FO sentence $\varphi$ with at most $r$ quantifiers that is true for every structure in $\cA$ but false for every structure in $\cB$. We call such a $\varphi$ a \emph{separating sentence} for $(\cA, \cB)$. In general, we wish to understand the number of quantifiers needed to express certain properties of structures (or in other words, separate structures having the property from those that do not).

\medskip \noindent{\bf Parallel Play: High-Level Idea. }In this paper, we introduce the concept of \emph{parallel play}, which widens the scope of winning strategies for $\bS$ compared to previous work. The essential idea is that if
the sets $\cA$ and $\cB$ can be partitioned as $\cA_1 \sqcup \ldots \sqcup \cA_k$ and $\cB_1 \sqcup \ldots \sqcup \cB_k$ respectively in such a way such that $\bS$ can win each instance $(\cA_i, \cB_i)$ of the MS ``sub-game'' by making his moves in these sub-games on the same sequence of sides, then he can win the entire instance $(\cA, \cB)$ with the same sequence of moves, provided that structures from different sub-games do not form a ``match'' at the end. In effect, $\bS$ plays the sub-games in parallel, thus saving many quantifiers in the resulting separating sentence.

\medskip \noindent{\bf Outline of the Paper. }
This paper is organized as follows. In Section \ref{sec:prelims}, we set up some preliminaries. In Section \ref{sec:parallel}, we precisely formulate what we call the \emph{Parallel Play Lemma} (Lemma \ref{lem:parallelplay}) and the \emph{Generalized Parallel Play Lemma} (Lemma \ref{lem:genparallelplay}). In Section \ref{sec:linearorders}, we develop results on linear orders that are important for the subsequent work on strings. 
One of our key results is that the number of quantifiers needed to separate linear orders up to a given length from bigger linear orders is at most one greater than the quantifier rank needed; this is remarkable in light of the fact that it is known (e.g.,~\cite{MScanon2}) that expressing a property in general can require exponentially more quantifiers than quantifier rank. We also show that we can distinguish any particular linear order by a sentence with essentially the optimal number of quantifiers (up to an additive term of $3$), which has strictly alternating quantifiers ending with a $\forall$.
In Section \ref{sec:strings}, we present our results on separating disjoint sets of strings. Table \ref{table:upper-bounds} provides a summary of our upper bounds for separating different-sized sets of strings. One of our more remarkable results (Theorem \ref{thm:all-vs-all upper bound}) is that we can separate any two arbitrary disjoint sets of $n$-bit strings with $(1 + \varepsilon)\frac{n}{\log(n)}$ quantifiers, for arbitrarily small $\varepsilon > 0$. In Proposition \ref{prop:all-vs-all lower bound}, we provide an $\frac{n}{\log(n)}$ lower bound, so the upper bound is essentially tight. We urge the reader to use Table \ref{table:upper-bounds} as a navigational guide when reading this section. In Section \ref{sec:conclusion}, we wrap up with some conclusions and open problems.

\begin{table*}[ht]
\centering
\resizebox{\textwidth}{!}{%
\begin{tabular}{ |c|c|c|c| } 
\hline
\textbf{$|\cA|$ vs.\ $|\cB|$} & \textbf{Quantifier Upper Bound} & 
\textbf{Proved In}
& \textbf{Range Where Bound is Best} \\ 
\hline 
1 vs.\ 1 & $\log(n) + O(1)$ & Proposition \ref{prop:one-vs-one-n}  & N/A\\
\hline
1 vs.\ $f(n)$ & \shortstack{$\log(n) + \log_t(N) + O(1)$ \\for $t = \max(2, \arg\max(t^{et} \leq N))$} & Theorem \ref{thm:one-vs-constant-best} & $f(n)$ is bounded, e.g., $f(n) \leq N$ \\
\hline
1 vs.\ $f(n)$  & $(1 + \varepsilon)\log(n) +  O(1)$ & Corollary \ref{cor:1+epsilon-upper} & $f(n)$ is polynomial \\ 
\hline
1 vs.\ $f(n)$ & $3\log_3(n) + O(1)$ & Theorem \ref{thm:one-vs-all-n}  & $f(n)$ is super-polynomial \\
\hline
$f(n)$ vs.\ $g(n)$ & $(1 + \varepsilon)\log(n) + O(1)$ & Corollary \ref{cor:1+epsilon-upper-2-sided} & $f(n)$, $g(n)$ are polynomials\\
\hline
$f(n)$ vs.\ $g(n)$ & $(3 + \varepsilon)\log_3(n) + O(1)$ & Corollary \ref{cor:3+epsilon-upper} & \shortstack{$f(n)$ is polynomial,\\ $g(n)$ is super-polynomial}\\
\hline  
$f(n)$ vs.\ $g(n)$ & $(1 + \varepsilon)\frac{n}{\log(n)} + O(1)$  & Theorem \ref{thm:all-vs-all upper bound} & $f(n)$, $g(n)$ are super-polynomial\\
\hline
\end{tabular}}
\caption{Ranges of applicability of our upper bounds for separating different sized sets $\cA$ and $\cB$ of strings of length $n$. Bounds involving $\varepsilon$ hold for arbitrarily small $\varepsilon > 0$. The results in rows one and four also apply when the string(s) on the right have arbitrary lengths, by Corollaries \ref{cor:one-vs-one} and \ref{cor:one-vs-all}. 
}
\label{table:upper-bounds}
\end{table*}
\section{Preliminaries}\label{sec:prelims}

Fix a vocabulary $\tau$ with finitely many relation and constant symbols. We typically designate structures in boldface ($\bA$), their universes in capital letters ($A$), and sets of structures in calligraphic typeface ($\cA$). This last convention includes sets of pebbled structures (see below).

We always use $\log(\cdot)$ to designate the base-$2$ logarithm. Furthermore, in several results in Section \ref{sec:strings}, we have an $O(1)$ additive term. This term will always be independent of $n$. Any additional dependence will be stated in the form of a subscript on the $O$, e.g., $O_t(1)$ would denote a term independent of $n$, but dependent on the choice of some parameter $t$.

\subsection{Pebbled Structures and Matching Pairs}

Consider a palette $\mathcal{C} = \{\r, \b, \g, \ldots\}$ of \emph{pebble colors}, with infinitely many pebbles of each color available. A $\tau$-structure $\bA$ is \emph{pebbled} if some of its elements $a_1, a_2, \ldots \in A$ have pebbles on them. There can only be at most one pebble of each color on a pebbled structure. There can be multiple pebbles (of different colors) on the same element $a_i \in A$. Occasionally, when the context is clear, we will use the term \emph{board} synonymously with ``pebbled structure''.
  
If $\mathbf{A}$ is a $\tau$-structure, and the first few pebbles are placed (in order) on elements $a_1, a_2, a_3 \ldots \in A$, we designate the resulting pebbled $\tau$-structure as $\langle \mathbf{A} ~|~ a_1, a_2, a_3, \ldots \rangle$. Note that $\bA$ can be viewed as a pebbled structure $\langle \bA ~|~\rangle$ with the empty set of pebbles.

By convention, we use $\r$, $\b$, and $\g$ for the first three pebbles we play (in that order), as a visual aid in our proofs. Hence, the pebbled structure $\langle \bA ~|~ a_1, a_2, a_3\rangle$ has pebbles $\r$ on $a_1 \in A$, $\b$ on $a_2 \in A$, and $\g$ on $a_3 \in A$. Note that $a_1$, $a_2$, and $a_3$ need not be distinct.


Let $\bA, \bB$ be $\tau$-structures, and let $a_1, \ldots, a_k \in A$ and $b_1, \ldots, b_k \in B$. We say that the pebbled structures $\langle \bA ~|~ a_1, \ldots, a_k\rangle$ and $\langle \bB ~|~ b_1, \ldots, b_k\rangle$ are a \emph{matching pair} if the map $f : A \to B$ defined by:
\begin{itemize}
    \item $f(a_i) = b_i$ for all $1 \leq i \leq k$
    \item $f(c^\bA) = c^\bB$ for all constants $c$ in $\tau$
\end{itemize}
is an isomorphism on the induced substructures. Note that $\langle \bA ~|~ a_1, \ldots, a_k\rangle$ and $\langle \bB ~|~ b_1, \ldots, b_k\rangle$ can form a matching pair even when $\bA \not\cong \bB$.


\subsection{Multi-Structural Games}

Let $r \in \mathbb{N}$, and let $\cA$ and $\cB$ be two sets of pebbled structures, each pebbled with the \emph{same} set $\{x_1, \ldots, x_k\} \subseteq \cC$ of pebbles. The \emph{$r$-round multistructural (MS) game on $(\cA, \cB)$} is defined as the following two-player game, played by two perfect players, \textbf{Spoiler} ($\bS$, he/him) and \textbf{Duplicator} ($\bD$, she/her). In round $i$ for $1 \leq i \leq r$, $\bS$ chooses either $\cA$ or $\cB$, and an \textbf{unused} color $y_i \in \mathcal{C}$; he then places (``plays'') a pebble of color $y_i$ on an element of \emph{every} board in the chosen set (``side''). In response, $\bD$ makes as many copies as she wants of each board on the other side, and plays a pebble of color $y_i$ on an element of each of those boards. $\bD$ wins the game if at the end of round $r$, there is a board in $\cA$ and a board in $\cB$ forming a matching pair. Otherwise, $\bS$ wins.
For readability, we always call the two sets $\cA$ and $\cB$, even though the structures change over the course of a game in the following two ways:
\begin{itemize}
    \item $\cA$ or $\cB$ can increase in size over rounds, as $\bD$ can make copies of the boards.
    \item The number of pebbles on each of the boards in $\cA$ and $\cB$ increases by $1$ in each round.
\end{itemize}

We will usually refer to $\cA$ as the \emph{left} side, and $\cB$ as the \emph{right} side.

Let $\cA$ and $\cB$ be two sets of pebbled structures, with each pebbled structure containing pebbles $\{x_1, \ldots, x_k\} \subseteq \mathcal{C}$. Let $\varphi(x_1, \ldots, x_k)$ be an FO formula with free variables $\{x_1, \ldots, x_k\}$. We say $\varphi$ is a \emph{separating formula} for $(\cA, \cB)$ (or $\varphi$ \emph{separates} $\cA$ and $\cB$) if:
\begin{itemize}
    \item for every $\langle \bA ~|~ a_1, \ldots, a_k\rangle \in \cA$ we have $\bA[a_1/x_1, \ldots, a_k/x_k] \models \varphi$,
    \item for every $\langle\bB ~|~ b_1, \ldots, b_k\rangle \in \cB$ we have $\bB[b_1/x_1, \ldots, b_k/x_k] \models \lnot\varphi$.
\end{itemize}
The following key theorem \cite{MScanon0, MScanon1}, stated here without proof, relates the logical characterization of a separating formula with the combinatorial property of an MS game strategy.

\begin{theorem}[Fundamental Theorem of MS Games, \cite{MScanon0, MScanon1}]\label{thm:MSfundamental}
    $\bS$ has a winning strategy in the $r$-round MS game on $(\cA, \cB)$ iff there is a formula with $\leq r$ quantifiers separating $\cA$ and $\cB$.
\end{theorem}

In the theorem above, if $\cA$ and $\cB$ are sets of \emph{unpebbled} structures, and $\varphi$ is a sentence, we would call $\varphi$ a \emph{separating sentence} for $(\cA, \cB)$.

We note that $\bD$ has a clear optimal strategy in the MS game, called the \emph{oblivious} strategy: for each of $\bS$'s moves, $\bD$ can make all possible responses on each pebbled structure on the other side (making enough copies of each such pebbled structure to be able to make these responses). For any instance of the MS game, if $\bD$ has a winning strategy, then the oblivious strategy is winning. For this reason, the MS game is essentially a single-player game, where $\bS$ can simulate $\bD$'s responses himself.

We make an easy observation here without proof, that will help us \emph{discard} some boards during gameplay; we can remove them without affecting the result of the game. This will help us in the analysis of several results in the paper.
\begin{observation}\label{obs:discard}
During gameplay in any instance of the MS game, consider a board $\langle \bA ~|~ a_1, \ldots, a_k\rangle$ such that there is no board on the other side forming a matching pair with it. Then, $\langle \bA ~|~ a_1, \ldots, a_k\rangle$ can be removed from the game without affecting the result.
\end{observation}

\subsection{Linear Orders}

Let $\tau_\mathsf{ord} = \angle{< ~;~ \mathsf{min},\mathsf{max}}$ be the vocabulary of orders, where $<$ is a binary predicate, and $\mathsf{min}$ and $\mathsf{max}$ are constant symbols. For every $\ell \geq 1$, we shall use $L_\ell$ to refer to a structure of type $\tau_\mathsf{ord}$, which interprets $<$ as a total linear order on $\ell + 1$ elements, and $\mathsf{min}$ and $\mathsf{max}$ as the first and last elements in that total order respectively. Note that there is only one linear order for any fixed value of $\ell$. Whenever unambiguous, we may suppress the subscript and refer to the linear order as simply $L$.

We define the \emph{length} of a linear order $L$ as the size of its universe minus one (equivalently, as the number of edges if the linear order were represented as a path graph). Hence, the length of $L_\ell$ is $\ell$. Since we only consider $\ell \geq 1$, the length is always positive; in particular, $\mathsf{min}$ and $\mathsf{max}$ are necessarily distinct. Our conventions are somewhat different from \cite{MScanon1} and \cite{MScanon2}, where the length of a linear order was the number of elements, and the vocabulary had no built-in constants. Note that having $\mathsf{min}$ and $\mathsf{max}$ is purely for convenience; each can be defined and reused at the cost of two quantifiers.

Let $L$ be a linear order with elements $a$ and $b$ satisfying $a < b$. The linear order $L[a, b]$ is the induced linear order on all elements from $a$ to $b$, both inclusive. If the variables $x$ and $y$ have been interpreted by $L$ so that $x^L = a$ and $y^L = b$, then we shall use $L[x, y]$ and $L[a, b]$ interchangeably; we adopt a similar convention for constants. If pebbles $\r$ and $\b$ have been placed on $L$ on $a$ and $b$ respectively, we use $L[\r, \b]$ and $L[a, b]$ interchangeably. Note that it is always the case that $L = L[\mathsf{min}, \mathsf{max}]$.

We will frequently need to consider sets of linear orders. For $\ell \geq 1$, we will use the notation $L_{\leq \ell}$ to denote the set of linear orders of length at most $\ell$, and $L_{> \ell}$ to denote the set of linear orders of length greater than $\ell$.

\subsection{Strings}

Let $\tau_\mathsf{string} = \langle <, ~ S ~;~ \mathsf{min},\mathsf{max}\rangle$ be the vocabulary of binary strings, where $<$ is a binary predicate, $S$ is a unary predicate, and $\mathsf{min}$ and $\mathsf{max}$ are constant symbols. We encode a string $w = (w_1, \ldots, w_n) \in \{0,1\}^n$ by the $\tau_\mathsf{string}$-structure $\mathbf{B}_w$ having universe $B_w = \{1, \dots, n\}$, relation $<$ interpreted by the linear order on $\{1, \dots, n\}$, relation $S = \{ i ~|~ w_i = 1 \} $, and $\mathsf{min}$ and $\mathsf{max}$ interpreted as $1$ and $n$ respectively.

For any $n$-bit string $w$, and any $1 \leq i \leq j \leq n$, denote by $w[i, j]$ the substring $w_i\ldots w_j$ of $w$. Note that $w[i, j]$ corresponds to the induced substructure of $\mathbf{B}_w$ on $\{i, \ldots, j\}$. Of course, $w = w[1, n]$. We will often interchangeably talk about the string $w$ and the $\tau_\mathsf{string}$-structure $B_w$, when the context is clear. Note that, as in $\tau_{\mathsf{ord}}$, having $\mathsf{min}$ and $\mathsf{max}$ in the vocabulary is purely for convenience.

\section{Parallel Play}\label{sec:parallel}

In this section, we prove our key lemma, which shows how, in certain cases, $\bS$ can combine his winning strategies in two sub-games, playing them in parallel in a single game that requires no more rounds than the longer of the two sub-games.

To understand why this is helpful, note that in general, if a formula $\varphi$ is of the form $\varphi_1\land\varphi_2$ or $\varphi_1\lor\varphi_2$, the number of quantifiers in $\varphi$ is the sum of the number of quantifiers in $\varphi_1$ and $\varphi_2$, even if the two subformulas have the same quantifier structure. We will see that playing parallel sub-games roughly corresponds to taking a $\varphi$ of the form $\varphi_1\land\varphi_2$ or $\varphi_1\lor\varphi_2$ where the subformulas have the same quantifier signature, and writing $\varphi$ with the same quantifier signature as $\varphi_1$ or $\varphi_2$, saving half the quantifiers we normally require.

Suppose $\bS$ has a winning strategy for an instance $(\cA, \cB)$ of the $r$-round MS game. In principle, the choice of which side $\bS$ plays on could depend on $\bD$'s previous responses. However, note that any strategy $\cS$ used by $\bS$ that wins against the oblivious strategy also wins against any other strategy that $\bD$ plays. Therefore, WLOG we may restrict ourselves to strategies used by $\bS$ against $\bD$'s oblivious strategy. It follows that the choice of which side to play on in every round is completely determined by the instance, independent of any of $\bD$'s responses. Let $\cS$ be such a winning strategy for $\bS$. We now define the \emph{pattern} of $\cS$, which specifies which side $\bS$ plays on in each round, when following $\cS$.

\begin{definition}\label{def:pattern}
Suppose $\cA$ and $\cB$ are sets of pebbled structures, and assume that $\bS$ has a winning strategy $\cS$ for the $r$-round $\ms$ game on $(\cA, \cB)$. The \emph{pattern} of $\cS$, denoted $\mathsf{pat}(\cS)$, is 
an $r$-tuple $(Q_1, \ldots, Q_r) \in \{\exists, \forall\}^r$, where:
\begin{equation*}
    Q_i = \begin{cases}
    \exists  & \text{ if $\bS$ plays in $\cA$ in round $i$,} \\
    \forall  & \text{ if $\bS$ plays in $\cB$ in round $i$.}
\end{cases}
\end{equation*}
We say that $\bS$ \emph{wins the game with pattern
$(Q_1,\ldots,Q_r)$} if $\bS$ has a winning strategy $\cS$ for the game in which $\mathsf{pat}(\cS) = (Q_1,\ldots,Q_r)$.
\end{definition}  

The following lemma is implicit in the proof of Theorem \ref{thm:MSfundamental}.

\begin{lemma}\label{lem:pattern}
For any two sets $\cA$ and $\cB$ of pebbled $\tau$-structures, the following are equivalent:
\begin{enumerate}
\item $\bS$ wins the $r$-round $\ms$ game on $(\cA, \cB)$ with pattern $(Q_1,\ldots,Q_r)$.
\item $(\cA,\cB)$ has a separating formula with $r$ quantifiers and quantifier signature $(Q_1,\ldots,Q_r)$.
\end{enumerate}
\end{lemma}

Note that Lemma \ref{lem:pattern} implies that, as long as there is a separating formula $\varphi$ for $(\cA, \cB)$ with $r$ quantifiers, $\bS$ has a winning strategy for the $r$-round MS game on $(\cA, \cB)$ that ``follows'' $\varphi$; namely, if $\varphi = Q_1\ldots Q_r\psi$, then in round $i$, $\bS$ plays in $\cA$ if $Q_i = \exists$, and in $\cB$ if $Q_i = \forall$. Hence, for the rest of the paper, we will refer to $\bS$ moves in $\cA$ and $\cB$ as \emph{existential} and \emph{universal} moves respectively. We are now ready to state our main lemma from this section.

\begin{lemma}[Parallel Play Lemma --- Baby Version]\label{lem:parallelplaybaby}
Let  $\cA$ and $\cB$ be two sets of pebbled structures, and let $r \in \mathbb{N}$. Suppose that $\cA$ and $\cB$ can be partitioned as $\cA = \cA_1 \sqcup \cA_2$ and $\cB = \cB_1 \sqcup \cB_2$
respectively, such that for $1 \leq i \leq 2$, $\bS$ has a winning strategy $\cS_i$ for the $r$-round MS game on $(\cA_i,\cB_i)$, satisfying:
\begin{enumerate}
\item Both $\cS_i$'s have the same pattern $P = \mathsf{pat}(\cS_1)=\mathsf{pat}(\cS_2)$.
\item At the end of the sub-games, both of the following are true:
\begin{itemize}
    \item There does not exist a board in $\cA_1$ and a board in $\cB_2$ forming a matching pair.
    \item There does not exist a board in $\cA_2$ and a board in $\cB_1$ forming a matching pair.
\end{itemize}
\end{enumerate}
Then $\bS$ wins the $r$-round MS game on $(\cA,\cB)$ with pattern $P$.
\end{lemma}

\begin{proof}
$\bS$ plays the $r$-round MS game on $(\cA,\cB)$ by playing his winning strategies, $\cS_i$ on $(\cA_i,\cB_i)$ simultaneously for all $i$ in parallel. This is a well-defined strategy, since both $\cS_i$'s have the same pattern $P$. At the end of the game, for $1 \leq i, j \leq 2$:
\begin{itemize}
\item For $i = j$, no board from $\cA_i$ forms a matching pair with a board from $\cB_j$, since $\bS$ wins the sub-game $(\cA_i, \cB_i)$.
\item For $i \neq j$, no board from $\cA_i$ forms a matching pair with a board from $\cB_j$, by assumption.
\end{itemize}
Therefore, no matching pair remains after round $r$, and so, $\bS$ wins the game. The pattern for this strategy is $P$ by construction.
\end{proof}

We observe that a useful generalization of Lemma \ref{lem:parallelplay} holds, namely that we can
weaken assumption 1 in the statement of the lemma, so that both of the patterns are subsequences of
a pattern $P$ of length $r$. This is because $\bS$ can simply extend $P_i$ for
strategy $\cS_i$ to a strategy $\cS'_i$ of pattern $P$, where for every ``missing'' entry $Q_i$ from $P$, $\bS$ makes a dummy move placing pebble $i$ on the corresponding side. Thus the two sub-games are both of length $r$ and have pattern $P$. We state this generalization below without a proof.

\begin{lemma}[Parallel Play Lemma]\label{lem:parallelplay}
Let  $\cA$ and $\cB$ be two sets of pebbled structures, and let $r \in \mathbb{N}$. Let $P \in \{\exists, \forall\}^r$ be a sequence of quantifiers of length $r$. Suppose that 
$\cA$ and $\cB$ can be partitioned as $\cA = \cA_1 \sqcup \cA_2$ and $\cB = \cB_1
\sqcup \cB_2$ respectively, such that for $1 \leq i \leq 2$, $\bS$ has a winning
strategy $\cS_i$ for the $r_i$-round MS game on $(\cA_i,\cB_i)$, with $r_i \leq r$, satisfying:
\begin{enumerate}
\item For $1\leq i\leq 2$, $\mathsf{pat}(\cS_i)$ is a subsequence of $P$.
\item At the end of the sub-games, for $i \neq j$, there does not exist a board in $\cA_i$ and a board in $\cB_j$ forming a matching pair.
\end{enumerate}
Then $\bS$ wins the $r$-round MS game on $(\cA,\cB)$ with pattern $P$.
\end{lemma}
\section{Linear Orders}\label{sec:linearorders}

Let $r(\ell)$ (resp.~$q(\ell)$) be the minimum quantifier rank (resp.~number of quantifiers) needed to separate $L_{\leq \ell}$ and $L_{>\ell}$. Let $q_\forall(\ell)$ (resp.~$q_\exists(\ell)$) be the minimum number of quantifiers needed to separate $L_{\leq \ell}$ and $L_{>\ell}$ with a sentence whose prenex normal form starts with $\forall$ (resp.~$\exists$). Note that $q(\ell) = \min(q_\forall(\ell),q_\exists(\ell))$. The
values of $r(\ell)$ are well understood \cite{ROSENSTEIN:1982}:

\begin{theorem}[Quantifier Rank, \cite{ROSENSTEIN:1982}]\label{rosey-fact}
For $\ell\geq 1$, we have $r(\ell) = 1 + \floor{\log(\ell)}$.
\end{theorem}\quad

Quantifier rank is a lower bound for number of quantifiers, so $r(\ell) \leq q(\ell)$ for all $\ell$. On the other hand, in \cite[Theorem 4]{MScanon2}, it is shown that the number of quantifiers needed to express a property can in general be super-exponentially more than the quantifier rank needed to express the same property. In this section, we shall show that this is not true of lengths: for each $\ell > 0$, we will show that $\bS$ can always separate $L_{\leq \ell}$ from $L_{>\ell}$ in at most $r(\ell) + 1$ rounds, which immediately shows that $q(\ell) \leq r(\ell) + 1$. In fact, we will also show that both $q_\forall(\ell)$ and $q_\exists(\ell)$ are bounded above by $r(\ell) + 1$.

Throughout this section, for notational convenience, we denote by $\textrm{MSL}_{\exists,r}(\ell)$ an $r$-round MS game on $(L_{\leq \ell}, L_{> \ell})$, in which $\bS$ \emph{must} play an existential first round move. We use $\textrm{MSL}_{\forall,r}(\ell)$ analogously, where the first round move \emph{must} be universal. Observe that, \emph{a priori}, for $Q \in \{\exists, \forall\}$, the $\textrm{MSL}_{Q,r}(\ell)$ may be winnable by either $\bS$ or $\bD$. Since we are primarily interested in analyzing games from $\bS$'s perspective, we restrict our attention only to $\bS$-winnable games. We call such games simply \emph{winnable}.

We now define a divide-and-conquer strategy for $\bS$ to play winnable instances $\textrm{MSL}_{Q,r}(\ell)$. We will start by building some intuition through an example game, before formally specifying the strategy. This strategy will implicitly give us upper bounds on $q_\exists(\ell)$ and $q_\forall(\ell)$, which we will then relate to $r(\ell)$. We will conclude the section with two useful results characterizing the quantifier pattern associated with the strategy, which will be used in Section \ref{sec:strings}.

For the rest of this section, for convenience, we define the \emph{closest-to-midpoint} of a linear order $L[x, y]$ as the element halfway between the elements corresponding to $x$ and $y$ if $L[x, y]$ has even length, or the one just to its left if $L[x, y]$ has odd length.

\subsection{A Strategy and an Example}\label{sec:cma}

We now develop a recursive strategy for $\bS$ for winnable instances $\textrm{MSL}_{Q,r}(\ell)$, called \emph{Closest-to-Midpoint with Alternation} ($\mathsf{CMA}$). The pattern for this strategy will alternate between $\exists$ and $\forall$, splitting each game recursively into two smaller sub-games that can be played in parallel using Lemma \ref{lem:parallelplay}. As we shall see, in these sub-games, placed pebbles will take on the roles of $\mathsf{min}$ and $\mathsf{max}$. $\bS$ will continue in this way until the sub-games are defined on linear orders of length $2$ or less, at which point he can win them easily.

The idea for the $\mathsf{CMA}$ strategy on $\textrm{MSL}_{Q,r}(\ell)$ is for $\bS$ to use the following rules throughout the game, except possibly the last three rounds:
\begin{itemize}
    \item $\bS$ starts on his designated side (determined by $Q$), and then alternates in every round;
    \item On every board, $\bS$ plays on the closest-to-midpoint of a linear order $L[x, y]$, chosen carefully to ensure he essentially ``halves'' the length of the instance every round.
\end{itemize}
Note that one consequence of the second point above is that $\bS$ will \textit{never} play on $\mathsf{max}$.

Before getting to a formal description of the strategy, let us illustrate the main idea through a worked example. Consider the (winnable) game $\textrm{MSL}_{\exists,4}(5)$.

In round $1$, $\bS$ plays on the closest-to-midpoint of all boards in $L_{\leq 5}$ (by the two conditions in the $\mathsf{CMA}$ strategy). Before $\bD$'s response, we reach the position shown in Figure \ref{fig:strategy_sample_game}.

\begin{figure*}
\begin{center}
\tikzset{Red/.style={circle,draw=red, inner sep=0pt, minimum size=.9cc, line width=2pt}}
\tikzset{RedOuter/.style={circle,draw=red, inner sep=0pt, minimum size=1.2cc, line width=2pt}}
\tikzset{BlueOuter/.style={circle,draw=blue, inner sep=0pt, minimum size=1.25cc, line width=2pt}}
\tikzset{Blue/.style={circle,draw=blue, inner sep=0pt, minimum size=.9cc, line width=2pt}}
\tikzset{Green/.style={circle,draw=dg, inner sep=0pt, minimum size=.9cc, line width=2pt}}
\tikzset{Black/.style={circle,draw=black, inner sep=0pt, minimum size=.9cc, line width=1pt}}
\tikzset{TreeNode/.style={rectangle,draw=black, inner sep=0pt, minimum size=1.3cc, line width=1pt}}
\tikzset{BigTreeNode/.style={rectangle,draw=black, inner sep=0pt, minimum size=1.8cc, line width=1pt}}
\tikzset{Tleft/.style={ellipse,draw=black, inner sep=0pt, minimum size=.9cc,
    line width=1pt,fill=orange!20}}
\tikzset{TRleft/.style={ellipse,draw=red, inner sep=0pt, minimum size=.9cc,
    line width=2pt,fill=orange!20}}
\tikzset{Tright/.style={ellipse,draw=black, inner sep=0pt, minimum size=.9cc,
    line width=1pt,fill=indigo!20}}

\begin{tikzpicture}[scale=.09]
\node at (-35,74) {$L_{\leq 5}$};
\node at (40,74) {$L_{> 5}$};
\node at (40,46) {$\vdots$};

\node [TreeNode] (T0) at(-10,64)  {$\exists \r$};
\node [TRleft] (a0) at (-40,64) {{\scriptsize \mn}};
\node [Tright] (a1) at (-30,64) {{\scriptsize \mx}};
\node [Tleft] (b0) at (10,64) {{\scriptsize \mn}};
\node [Black] (b1) at (20,64) {{\scriptsize $$}};
\node [Black] (b2) at (30,64) {{\scriptsize $$}};
\node [Black] (b3) at (40,64) {{\scriptsize $$}};
\node [Black] (b4) at (50,64) {{\scriptsize $$}};
\node [Black] (b5) at (60,64) {{\scriptsize $$}};
\node [Tright] (b6) at (70,64) {{\scriptsize \mx}};

\foreach \from/\to in {a0/a1,b0/b1,b1/b2,b2/b3,b3/b4,b4/b5,b5/b6}
\draw[-,line width=1pt,color=black] (\from) -- (\to);

\node [Tleft] (1a0) at (-45,58) {{\scriptsize \mn}};
\node [Red]   (1a1) at (-35,58) {{\scriptsize $\r$}};
\node [Tright] (1a2) at (-25,58) {{\scriptsize \mx}};
\node [Tleft] (1b0) at (5,58) {{\scriptsize \mn}};
\node [Black] (1b1) at (15,58) {{\scriptsize $$}};
\node [Black] (1b2) at (25,58) {{\scriptsize $$}};
\node [Black] (1b3) at (35,58) {{\scriptsize $$}};
\node [Black] (1b4) at (45,58) {{\scriptsize $$}};
\node [Black] (1b5) at (55,58) {{\scriptsize $$}};
\node [Black] (1b6) at (65,58) {{\scriptsize $$}};
\node [Tright] (1b7) at (75,58) {{\scriptsize \mx}};
\foreach \from/\to in {1a0/1a1,1a1/1a2,1b0/1b1,1b1/1b2,1b2/1b3,1b3/1b4,1b4/1b5,1b5/1b6,1b6/1b7}
\draw[-,line width=1pt,color=black] (\from) -- (\to);

\node [Tleft] (2a0) at (-50,52) {{\scriptsize \mn}};
\node [Red] (2a1) at (-40,52) {{\scriptsize $\r$}};
\node [Black] (2a2) at (-30,52) {{\scriptsize $$}};
\node [Tright] (2a3) at (-20,52) {{\scriptsize \mx}};
\foreach \from/\to in {2a0/2a1,2a1/2a2,2a2/2a3}
\draw[-,line width=1pt,color=black] (\from) -- (\to);

\node [Tleft] (3a0) at (-55,46) {{\scriptsize \mn}};
\node [Black] (3a1) at (-45,46) {{\scriptsize $$}};
\node [Red] (3a2) at (-35,46) {{\scriptsize $\r$}};
\node [Black] (3a3) at (-25,46) {{\scriptsize $$}};
\node [Tright] (3a4) at (-15,46) {{\scriptsize \mx}};
\foreach \from/\to in {3a0/3a1,3a1/3a2,3a2/3a3,3a3/3a4}
\draw[-,line width=1pt,color=black] (\from) -- (\to);

\node [Tleft] (4a0) at (-60,40) {{\scriptsize \mn}};
\node [Black] (4a1) at (-50,40) {{\scriptsize $$}};
\node [Red] (4a2) at (-40,40) {{\scriptsize $\r$}};
\node [Black] (4a3) at (-30,40) {{\scriptsize $$}};
\node [Black] (4a4) at (-20,40) {{\scriptsize $$}};
\node [Tright] (4a5) at (-10,40) {{\scriptsize \mx}};
\foreach \from/\to in {4a0/4a1,4a1/4a2,4a2/4a3,4a3/4a4,4a4/4a5}
\draw[-,line width=1pt,color=black] (\from) -- (\to);
\end{tikzpicture}
\end{center}
\caption{The position after $\bS$'s round $1$ move in the game $\textrm{MSL}_{\exists,4}(5)$. The pebble $\r$ is on the closest-to-midpoint of every board on the left.}
\label{fig:strategy_sample_game}
\end{figure*}
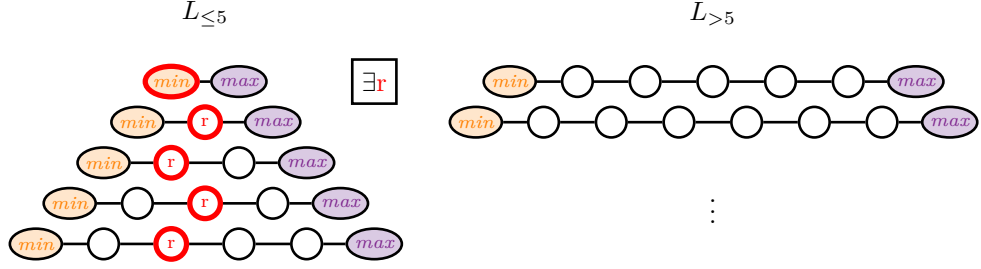  

Now assume $\bD$ responds obliviously. We can first use Observation \ref{obs:discard} to discard all boards on the right with $\r$ on $\mathsf{max}$. By virtue of $\bS$'s first move, every board $\langle L ~|~ a_1\rangle$ on the left satisfies \emph{both} $L[\mathsf{min}, \r] \leq 2$, and $L[\r, \mathsf{max}] \leq 3$. Now consider any board $\langle L' ~|~ a'_1\rangle$ on the right. Note that either $L'[\mathsf{min}, \r] > 2$, or $L'[\r, \mathsf{max}] > 3$. Partition the right side as $\cB_1 \sqcup \cB_2$, where every $\langle L' ~|~ a'_1\rangle \in \cB_1$ satisfies $L'[\mathsf{min}, \r] > 2$, and every $\langle L' ~|~ a'_1\rangle \in \cB_2$ satisfies $L'[\r, \mathsf{max}] > 3$.

In round $2$, $\bS$ makes a universal move (by the first condition in the $\mathsf{CMA}$ strategy). In all boards in $\cB_1$, he plays pebble $\b$ on the closest-to-midpoint of $L'[\mathsf{min}, \r]$; similarly, in all boards in $\cB_2$, he plays pebble $\b$ on the closest-to-midpoint of $L'[\r, \mathsf{max}]$. Note that in either case, $\bS$ plays $\b$ on an element which is not on $\r$, $\mathsf{min}$, or $\mathsf{max}$.

After $\bD$ responds obliviously, we can use Observation \ref{obs:discard} to discard all boards on the left where $\b$ is on $\mathsf{min}$, $\mathsf{max}$, or $\r$. Since in particular this discards all boards on the left with $\r$ on $\mathsf{min}$, we can again use Observation \ref{obs:discard} to discard all boards from the right which have $\r$ on $\mathsf{min}$. Every remaining board in $\cB_1$ (resp.~$\cB_2$) corresponds to the isomorphism class $\mathsf{min} < \b < \r < \mathsf{max}$ (resp.~$\mathsf{min} < \r < \b < \mathsf{max}$). The remaining boards on the left also correspond to exactly one of those classes. Partition the left side as $\cA_1 \sqcup \cA_2$ accordingly.

Now, because of this difference in isomorphism classes, we will never obtain a matching pair from $\cA_1$ and $\cB_2$ (or from $\cA_2$ and $\cB_1$). Furthermore, for the rest of the game, $\bS$ will \emph{only} play inside $L[\mathsf{min}, \r]$ on all boards in $\cA_1$ and $\cB_1$, and inside $L[\r, \mathsf{max}]$ on all boards in $\cA_2$ and $\cB_2$. Suppose, in response to such a move on $\cA_1$, $\bD$ plays outside the range $L[\mathsf{min}, \r]$ on a board from $\cB_1$; the resulting board cannot form a partial match with any board from $\cA_1$ (since there is a discrepancy with $\r$), or with any board from $\cA_2$ (as observed already). Therefore, this board from $\cB_1$ can be discarded using Observation \ref{obs:discard}. A similar argument applies if $\bD$ ever responds outside the corresponding range in $\cB_2$, $\cA_1$, or $\cA_2$.

It follows that the sub-game $(\cA_1, \cB_1)$ (resp.~$(\cA_2, \cB_2$) corresponds \emph{exactly} to the game $\textrm{MSL}_{\forall, 3}(2)$ (resp.~$\textrm{MSL}_{\forall, 3}(3)$) where $\bS$ has already made his first move using the $\mathsf{CMA}$ strategy by playing a universal move on the closest-to-midpoints of the (relevant) linear orders. Since $\bS$ will alternate sides throughout, the patterns for both sub-game strategies will be the same.

We can now apply Lemma \ref{lem:parallelplay}. Observe that the lengths of the instances in the sub-games have been roughly halved, at the cost of a single move. The game then proceeds as shown in Figure \ref{fig:small_game_tree}. The leaves of the tree correspond to base cases (analyzed in Section \ref{sec:cmaformal}). The pattern of the strategy is preserved along all branches.
\begin{figure}[ht]
    \centering
    \includegraphics[scale=0.6]{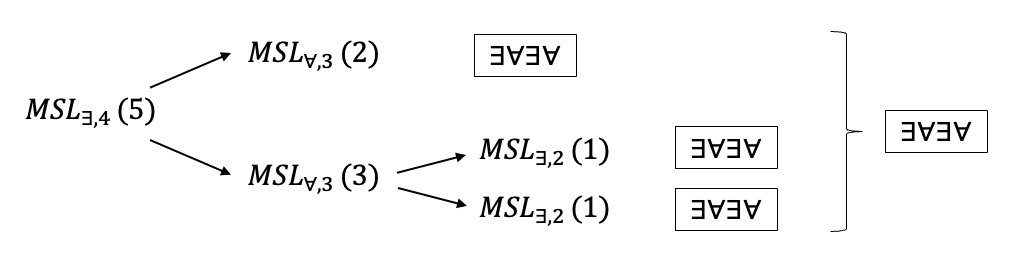}
    \caption{The $\textrm{MSL}_{\exists,4}(5)$ game tree. Each leaf is decorated with the associated quantifier signature. All paths can be played in parallel using Lemma \ref{lem:parallelplay} using the pattern $(\exists, \forall, \exists, \forall)$.}
    \label{fig:small_game_tree}
\end{figure}

\subsection{Formalizing the Strategy}\label{sec:cmaformal}

The first step in formalizing the $\mathsf{CMA}$ strategy for $\bS$ is to define four base cases, which we shall call \textit{irreducible} games. We assert the following.
\begin{enumerate}
    \item The game $\textrm{MSL}_{\forall, 1}(1)$ is winnable; $\bS$ makes a universal move by playing on any element other than $\mathsf{min}$ and $\mathsf{max}$ on each board on the right. There is no valid response by $\bD$ on the single board on the left. The pattern is $(\forall)$.
    \item The game $\textrm{MSL}_{\exists, 2}(1)$ is winnable; $\bS$ makes a dummy existential move (by playing as a matter of convention on $\mathsf{min}$), and then reverts to the strategy above for $\textrm{MSL}_{\forall, 1}(1)$ for his second move. The pattern is $(\exists, \forall)$.
    \item The game $\textrm{MSL}_{\forall, 2}(2)$ is winnable; $\bS$ makes two successive universal moves by playing on two arbitrary distinct elements other than $\mathsf{min}$ and $\mathsf{max}$ on each board on the right. $\bD$ cannot match this one the boards on the left. We remark that $\textrm{MSL}_{\forall, 2}(2)$ is not winnable by $\bS$ if he plays in any other fashion. The pattern is $(\forall, \forall)$.
    \item The game $\textrm{MSL}_{\forall, 3}(2)$ is winnable; $\bS$ follows the same strategy as in $\textrm{MSL}_{\forall, 2}(2)$ in rounds $1$ and $3$, except that he makes a dummy existential move in round $2$ (by playing as a matter of convention on $\mathsf{min}$). The pattern is $(\forall, \exists, \forall)$.
\end{enumerate}

The game $\textrm{MSL}_{\exists, 1}(1)$ is not winnable and hence not considered.

We now give a formalization of the inductive step. For a given quantifier $Q \in \{\exists, \forall\}$ and its complementary quantifier $\bar{Q}$, consider the game $\textrm{MSL}_{Q,k}(\ell)$. Note that if $\bS$ employs the $\mathsf{CMA}$ strategy described in Section \ref{sec:cma}, the game splits\footnote{Recall that technically the split happens only after two rounds, but each sub-game then ends up with one move already played consistent with the strategy.} into the two sub-games $\textrm{MSL}_{\bar{Q},k-1}(\ell')$ and $\textrm{MSL}_{\bar{Q},k-1}(\ell'')$. We designate this split as:
\begin{equation*}
    \textrm{MSL}_{Q,k}(\ell) \rightarrow  \textrm{MSL}_{\bar{Q},k-1}(\ell') \oplus \textrm{MSL}_{\bar{Q},k-1}(\ell'').
\end{equation*}
We will shortly show in the proof of Lemma \ref{lem:cma-is-well-defined} that these sub-games can be played recursively, in parallel. When $\bS$ reaches an irreducible sub-game, he plays the strategies with the patterns asserted above.

We now claim the following about the rules for splitting.

\begin{restatable}[Splitting Rules]{claim}{splitrules}\label{claim:splitrules}
    For $k \geq 3$, we have:
\begin{eqnarray} \label{msl-cases}
    \textrm{(i)~MSL}_{\exists,k}(2\ell) \rightarrow  \textrm{MSL}_{\forall,k-1}(\ell) \oplus \textrm{MSL}_{\forall,k-1}(\ell),~~~~~~~~~~~~~~~~~~~\ell \geq 1 \notag\\
    \textrm{(ii)~MSL}_{\exists,k}(2\ell + 1) \rightarrow  \textrm{MSL}_{\forall,k-1}(\ell) \oplus \textrm{MSL}_{\forall,k-1}(\ell+1),~~~~~~~~~\ell \geq 1 \\
    \textrm{(iii)~MSL}_{\forall,k}(2\ell) \rightarrow \textrm{MSL}_{\exists,k-1}(\ell) \oplus \textrm{MSL}_{\exists,k-1}(\ell-1),~~~~~~~~~~~~~~\ell \geq 2 \notag \\
    \textrm{(iv)~MSL}_{\forall,k}(2\ell + 1) \rightarrow  \textrm{MSL}_{\exists,k-1}(\ell) \oplus \textrm{MSL}_{\exists,k-1}(\ell),~~~~~~~~~~~~~~\ell \geq 1 \notag
\end{eqnarray}
\end{restatable}

\begin{proof}
We prove Claim \ref{claim:splitrules} in cases (i) and (iii). The other two cases are similar. Note that our analysis of $\textrm{MSL}_{\exists,4}(5)$ was an example of case (ii).

Consider case (i). Figure \ref{parallel path case1 fig} shows the gameplay through to the configuration immediately after $\bS$'s round $2$ move. Once $\bD$ makes her oblivious response, we can discard some of the boards following Observation \ref{obs:discard}, and note that the game splits into two games: $(\cA_1, \cB_1)$ corresponding to the isomorphism class $\mathsf{min} < \r < \b < \mathsf{max}$, and $(\cA_2, \cB_2)$ corresponding to the isomorphism class $\mathsf{min} < \b < \r < \mathsf{max}$. The game proceeds within the linear orders $L[\r, \mathsf{max}]$ in $(\cA_1, \cB_1)$, and within the linear orders $L[\mathsf{min}, \r]$ in $(\cA_2, \cB_2)$, using Observation \ref{obs:discard} to discard any responses outside those ranges. By construction, these both correspond to $\textrm{MSL}_{\forall,k-1}(\ell)$ games.

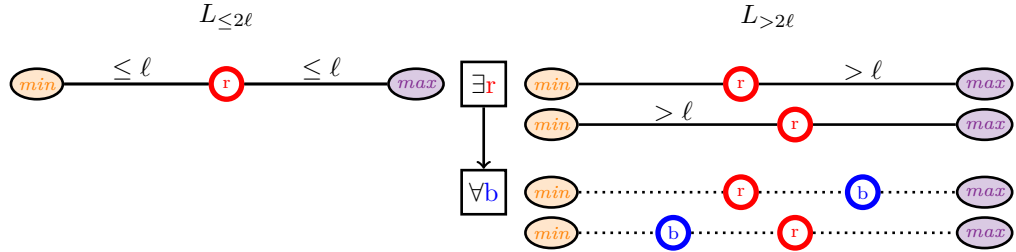
\begin{figure*}[ht]
\begin{center}
\tikzset{Tleft/.style={ellipse,draw=black, inner sep=0pt, minimum size=.9cc,
    line width=1pt,fill=orange!20}}
\tikzset{TRleft/.style={ellipse,draw=red, inner sep=0pt, minimum size=.9cc,
    line width=2pt,fill=orange!20}}
\tikzset{Tright/.style={ellipse,draw=black, inner sep=0pt, minimum size=.9cc,
    line width=1pt,fill=indigo!20}}
  \tikzset{Red/.style={circle,draw=red, inner sep=0pt, minimum size=.9cc, line width=2pt}}
\tikzset{Blue/.style={circle,draw=blue, inner sep=0pt, minimum size=.9cc, line width=2pt}}
\tikzset{Green/.style={circle,draw=dg, inner sep=0pt, minimum size=.9cc, line width=2pt}}
\tikzset{Black/.style={circle,draw=black, inner sep=0pt, minimum size=.9cc, line width=1pt}}
\tikzset{TreeNode/.style={rectangle,draw=black, inner sep=0pt, minimum size=1.3cc, line width=1pt}}
\tikzset{BigTreeNode/.style={rectangle,draw=black, inner sep=0pt, minimum size=1.8cc, line width=1pt}}

\begin{tikzpicture}[scale=.09]
\node at (-38,50) {$L_{\leq 2\ell}$};
\node at (42,50) {$L_{> 2\ell}$};

\node [TreeNode] (T0) at(0,40)  {};
   \node at (T0) {$\exists \r$};
  \node [TreeNode] (T1) at (0,24) {$\forall \b$};
  \draw[line width=1pt,->,color=black] (T0) -- (T1);

\node at (-52,42.25) {$\leq \ell$};
\node at (-24,42.25) {$\leq \ell$};
\node at (56,42) {$>\ell$};
\node [Tleft] (1a1) at (-66,40) {{\scriptsize \mn}};
\node [Red] (1amid) at (-38,40) {{\scriptsize $\r$}};
\node [Tright] (1a8) at (-10,40) {{\scriptsize \mx}};
\node [Tright] (1b9) at (74,40) {{\scriptsize \mx}};
\node [Red]   (1b2) at (38,40) {{\scriptsize $\r$}};
\node [Tleft] (1b1) at (10,40) {{\scriptsize \mn}};
\foreach \from/\to in {1a1/1amid,1amid/1a8,1b1/1b2,1b2/1b9}
\draw[-,line width=1pt,color=black] (\from) -- (\to);

\node at (28,36) {$>\ell$};
\node [Tright] (1c9) at (74,34) {{\scriptsize \mx}};
\node [Red]   (1c2) at (46,34) {{\scriptsize $\r$}};
\node [Tleft] (1c1) at (10,34) {{\scriptsize \mn}};
\foreach \from/\to in {1a1/1amid,1amid/1a8,1c1/1c2,1c2/1c9}
\draw[-,line width=1pt,color=black] (\from) -- (\to);

\node [Tright] (2b9) at (74,24) {{\scriptsize \mx}};
\node [Red]   (2b2) at (38,24) {{\scriptsize $\r$}};
\node [Blue]   (2b3) at (56,24) {{\scriptsize $\b$}};
\node [Tleft] (2b1) at (10,24) {{\scriptsize \mn}};

\node [Tright] (2c9) at (74,18) {{\scriptsize \mx}};
\node [Red]   (2c2) at (46,18) {{\scriptsize $\r$}};
\node [Blue]  (2c3) at (28,18) {{\scriptsize $\b$}};
\node [Tleft] (2c1) at (10,18) {{\scriptsize \mn}};
\foreach \from/\to in {2b1/2b2,2b2/2b3,2b3/2b9,2c1/2c3,2c3/2c2,2c2/2c9}
\draw[dotted,line width=1pt,color=black] (\from) -- (\to);
\end{tikzpicture}
\end{center}
\caption{The first round and a half of the $\textrm{MSL}_{\exists,k}(2\ell)$ game
  according to the $\mathsf{CMA}$ strategy.}
\label{parallel path case1 fig}
\end{figure*}

Now consider case (iii). Figure \ref{path thm case2 fig} shows the gameplay through to the configuration immediately after $\bS$'s round $2$ move. Once $\bD$ makes her oblivious response, we can discard some of the boards following Observation \ref{obs:discard}, and note that the game splits into two games: $(\cA_1, \cB_1)$ corresponding to the isomorphism class $\mathsf{min} < \b < \r < \mathsf{max}$, and $(\cA_2, \cB_2)$ corresponding to the isomorphism class $\mathsf{min} < \r < \b < \mathsf{max}$. The game proceeds within the linear orders $L[\mathsf{min}, \r]$ in $(\cA_1, \cB_1)$, and within the linear orders $L[\r, \mathsf{max}]$ in $(\cA_2, \cB_2)$, using Observation \ref{obs:discard} to discard any responses outside those ranges. By construction, these correspond to an $\textrm{MSL}_{\exists,k-1}(\ell - 1)$ game and an $\textrm{MSL}_{\exists,k-1}(\ell)$ respectively.

\begin{figure*}[ht]
\begin{center}
\tikzset{Tleft/.style={ellipse,draw=black, inner sep=0pt, minimum size=.9cc,
    line width=1pt,fill=orange!20}}
\tikzset{TRleft/.style={ellipse,draw=red, inner sep=0pt, minimum size=.9cc,
    line width=2pt,fill=orange!20}}
\tikzset{Tright/.style={ellipse,draw=black, inner sep=0pt, minimum size=.9cc,
    line width=1pt,fill=indigo!20}}
\tikzset{Red/.style={circle,draw=red, inner sep=0pt, minimum size=.9cc, line width=2pt}}
\tikzset{Blue/.style={circle,draw=blue, inner sep=0pt, minimum size=.9cc, line width=2pt}}
\tikzset{Green/.style={circle,draw=dg, inner sep=0pt, minimum size=.9cc, line width=2pt}}
\tikzset{Black/.style={circle,draw=black, inner sep=0pt, minimum size=.9cc, line width=1pt}}
\tikzset{TreeNode/.style={rectangle,draw=black, inner sep=0pt, minimum size=1.3cc, line width=1pt}}
\tikzset{BigTreeNode/.style={rectangle,draw=black, inner sep=0pt, minimum size=1.8cc, line width=1pt}}

\begin{tikzpicture}[scale=.09]
\node at (-38,50) {$L_{\leq 2\ell}$};
\node at (42,50) {$L_{> 2\ell}$};
\node [TreeNode] (T0) at(0,40)  {};

   \node at (T0) {$\forall \r$};
  \node [TreeNode] (T1) at (0,24) {$\exists \b$};
  \draw[line width=1pt,->,color=black] (T0) -- (T1);

\node at (56,42.25) {$>\ell$};
\node at (24,42.25) {$\geq\ell$};
\node at (-54,42.25) {$\leq \ell-1$};
\node [Tright] (1a1) at (74,40) {{\scriptsize \mx}};
\node [Red] (1amid) at (38,40) {{\scriptsize $\r$}};
\node [Tleft] (1a8) at (10,40) {{\scriptsize \mn}};
\node [Tright] (1b9) at (-10,40) {{\scriptsize \mx}};
\node [Red]   (1b2) at (-42,40) {{\scriptsize $\r$}};
\node [Tleft] (1b1) at (-66,40) {{\scriptsize \mn}};
\foreach \from/\to in {1a1/1amid,1amid/1a8,1b1/1b2,1b2/1b9}
\draw[-,line width=1pt,color=black] (\from) -- (\to);

\node at (-21,36.25) {$\leq\ell$};
\node [Tright] (1c9) at (-10,34) {{\scriptsize \mx}};
\node [Red]   (1c2) at (-32,34) {{\scriptsize $\r$}};
\node [Tleft] (1c1) at (-66,34) {{\scriptsize \mn}};
\foreach \from/\to in {1a1/1amid,1amid/1a8,1c1/1c2,1c2/1c9}
\draw[-,line width=1pt,color=black] (\from) -- (\to);

\node [Tright] (2b9) at (-10,24) {{\scriptsize \mx}};
\node [Red]   (2b2) at (-42,24) {{\scriptsize $\r$}};
\node [Blue]   (2b3) at (-54,24) {{\scriptsize $\b$}};
\node [Tleft] (2b1) at (-66,24) {{\scriptsize \mn}};

\node [Tright] (2c9) at (-10,18) {{\scriptsize \mx}};
\node [Red]   (2c2) at (-32,18) {{\scriptsize $\r$}};
\node [Blue]  (2c3) at (-21,18) {{\scriptsize $\b$}};
\node [Tleft] (2c1) at (-66,18) {{\scriptsize \mn}};
\foreach \from/\to in {2b1/2b3,2b3/2b2,2b2/2b9,2c1/2c2,2c2/2c3,2c3/2c9}
\draw[dotted,line width=1pt,color=black] (\from) -- (\to);
\end{tikzpicture}
\end{center}
\caption{The first round and a half of the $\textrm{MSL}_{\forall,k}(2\ell)$ game
  according to the $\mathsf{CMA}$ strategy.}

\label{path thm case2 fig}
\end{figure*}
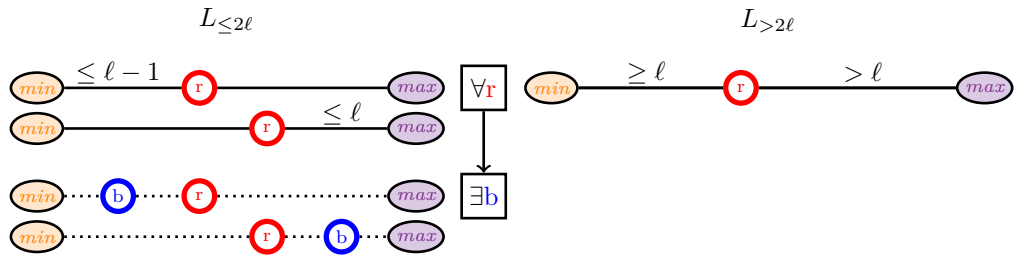

\end{proof}

Of course, the $\mathsf{CMA}$ strategy starts out seemingly promisingly, splitting with both initial sub-games starting on the same side; we must ensure that the strategy continues to be \emph{well-specified}, i.e., this continues throughout the recursion stack, especially since the sub-games can have different lengths. We show this in the following lemma.

\begin{restatable}{lemma}{cmawelldefined}\label{lem:cma-is-well-defined}
The $\mathsf{CMA}$ strategy is well-specified. Moreover, for $k \geq 3$, if $\textrm{MSL}_{Q,k}(\ell) \rightarrow  \textrm{MSL}_{\bar{Q},k-1}(\ell_1) \oplus \textrm{MSL}_{\bar{Q},k-1}(\ell_2)$ with $\ell_1 \geq \ell_2$, then the pattern of $\bS$'s winning strategy for $\textrm{MSL}_{Q,k}(\ell)$ is $Q$ concatenated with the pattern for the winning strategy for $\textrm{MSL}_{\bar{Q},k-1}(\ell_1)$.
\end{restatable}

\begin{proof}
We first make a simple claim without a proof.

\begin{claim} \label{claim:simple}
Suppose $\cA$ and $\cB$ are two sets of pebbled structures that can be partitioned into $\cA = \cA_1 \cup \cA_2$ and $\cB = \cB_1 \cup \cB_2$, such that no pebbled structure in $\cA_i$ forms a matching pair with a pebbled structure in $\cB_j$, for $i \neq j$. Suppose that $\bS$ has a strategy $\cS_1$ to win the $(r+3)$-round MS-game on $(\cA_1, \cB_1)$ with pattern $\mathsf{pat}(\cS_1) = (Q_1, \ldots, Q_r, \forall, \exists, \forall)$, and also has a strategy $\cS_2$ to win the $(r+2)$-round MS-game on $(\cA_2, \cB_2)$ with pattern $\mathsf{pat}(\cS_2) = (Q_1, \ldots, Q_k, \forall, \forall)$. Then, $\bS$ has a strategy $\cS$ to win the $(r+3)$-round MS-game on $(\cA, \cB)$ satisfying $\mathsf{pat}(\cS) = \mathsf{pat}(\cS_1)$.
\end{claim}
Note that the claim is immediate by Lemma \ref{lem:parallelplay}, since $\mathsf{pat}(\cS_2)$ is a subsequence of $\mathsf{pat}(\cS_1)$.

Back to the proof of the main lemma. Assume inductively that the lemma holds for all MSP games up to length $\ell$. Assume additionally that the pattern of both sub-games alternate, either ending in one or two universal quantifiers, and that the longer sub-game (if there is one) takes at most one more round than the shorter one. We shall show that both the lemma and this assumption continue to hold for the game of length $\ell + 1$.

There are four cases depending on whether we are considering an $\exists$ or $\forall$ game, and whether $\ell$ is even or odd.

\noindent \textbf{Case 1: $\exists$ game, $\ell$ odd.} The game under consideration is $\textrm{MSL}_{\exists, r}(\ell+1)$. Since $\ell + 1$ is even, we split into two sub-games that are both $\textrm{MSL}_{\forall, r-1}((\ell+1)/2)$, which clearly have the same pattern.

\noindent \textbf{Case 2: $\exists$ game, $\ell$ even.} The game under consideration is $\textrm{MSL}_{\exists, r}(\ell+1)$. Since $\ell+1$ is odd, we split into the sub-games $\textrm{MSL}_{\forall, r-1}(\ell/2)$ and $\textrm{MSL}_{\forall, r-1}(\ell/2 + 1)$. If the sub-games have strategies with the same number of moves, then they both must have the pattern $(\forall, \exists, \ldots, \forall, \exists, \forall)$ or $(\forall, \exists, \ldots, \forall, \forall)$. In this case, the lemma holds. Otherwise, the strategy for the $\textrm{MSL}_{\forall, r - 1}(\ell/2 + 1)$ game has one more round than the other game (by assumption), and then the two patterns for the sub-games can line up in one of the two following ways (in each case with the longer pattern on top):
\begin{align*}
    &\qquad\qquad (\forall, \exists, \cdots, \forall, \exists, \forall) && (\forall, \exists, \cdots, \forall, \exists, \forall, \forall) \\
    &\qquad\qquad (\forall, \exists, \cdots, \forall, \forall) && (\forall, \exists, \cdots, \forall, \exists, \forall)
\end{align*}
In the first case, we are done by Claim \ref{claim:simple}. In the second case, $\bS$ just plays an arbitrary universal move at the end of the second sub-game, and the lemma once again follows.

\noindent \textbf{Case 3: $\forall$ game, $\ell$ odd.} The game under consideration is $\textrm{MSL}_{\forall, r}(\ell+1)$. Since $\ell+1$ is even, we split into the sub-games $\textrm{MSL}_{\exists, r-1}((\ell+1)/2)$ and $\textrm{MSL}_{\exists, r-1}((\ell+1)/2 - 1)$. The analysis is now very similar to Case 2. If the sub-games have strategies with the same number of moves, then they both have the pattern $(\exists, \forall, \ldots, \exists, \forall)$ or $(\exists, \forall, \ldots, \exists, \forall, \forall)$, and the lemma follows. Otherwise, we have the possibilities:
\begin{align*}
    &\qquad\qquad (\exists, \forall, \cdots, \exists, \forall, \forall) && (\exists, \forall, \cdots, \exists, \forall, \exists, \forall) \\
    &\qquad\qquad (\exists, \forall, \cdots, \exists, \forall) && (\exists, \forall, \cdots, \exists, \forall, \forall)
\end{align*}
And the analysis of these two cases is just like in Case 2, whereby the lemma follows.

\noindent \textbf{Case 4:  $\forall$ game, $\ell$ even.} The game under consideration is $\textrm{MSL}_{\forall, r}(\ell+1)$. Since $\ell+1$ is odd, we split into two sub-games of $\textrm{MSL}_{\exists, r-1}(\ell/2)$, which have the same pattern.

The requisite alternation pattern is clearly maintained, from the definition of the strategy. The assumption that the longer sub-game takes at most one more round than the shorter sub-game follows by noting that the lengths of the patterns are monotonic in $\ell$ and never increase by more than one.
\end{proof}

\subsection{Bounding the Quantifier Number}\label{sec:boundingqn}

Let $q^*_\exists(\ell)$ (resp.~$q^*_\forall(\ell)$) be the minimum $r \in \mathbb{N}$ such that $\bS$ wins the game $\textrm{MSL}_{\exists,r}(\ell)$ (resp.~$\textrm{MSL}_{\forall,r}(\ell)$) using the $\mathsf{CMA}$ strategy. Of course, we must have $q_\exists(\ell) \leq q^*_\exists(\ell)$ and $q_\forall(\ell) \leq q^*_\forall(\ell)$. Let $q^*(\ell) = \min(q^*_\exists(\ell), q^*_\forall(\ell))$. The following lemma (whose proof is omitted) follows from the complete description of the strategy from Section \ref{sec:cmaformal}.

\begin{lemma}\label{lem:q*-vals}
We have $q^*_\forall(1) = 1$, $q^*_\exists(1) = 2$, and $q^*_\forall(2) = 2$. Also:
\begin{align*}
        &q^*_\exists(2\ell) = q^*_\forall(\ell) + 1~~~~\textrm{for }\ell \geq 1, 
        &&&q^*_\exists(2\ell+1) = q^*_\forall(\ell+1) + 1~~~~\textrm{for }\ell \geq 1, \\
        &q^*_\forall(2\ell) = q^*_\exists(\ell) + 1~~~~\textrm{for }\ell \geq 2,
        &&&q^*_\forall(2\ell+1) = q^*_\exists(\ell) + 1~~~~~~~~~\textrm{for }\ell \geq 1.
\end{align*}
\end{lemma}

From Lemma \ref{lem:q*-vals} it is possible to recursively compute $q^*_\forall(\ell)$ and $q^*_\exists(\ell)$, and therefore $q^*(\ell)$ for all values of $\ell \geq 1$. These values are provided for $\ell \leq 127$ in Table \ref{q*-table}.

\begin{table}[ht] 
\begin{center}
\begin{tabular}{|c|c|c|c|c|}\hline 
$\ell$ & $q^*_\forall(\ell)$ & $q^*_\exists(\ell)$ & $q^*(\ell)$ & $r(\ell)$ \\ \hline \hline
  1 & 1 & 2& 1 & 1\\ \hline
  2 & 2 & 2 & 2 & 2 \\ \hline
  3  & 3 & 3 & 3 & 2\\ \hline
  4 & 3 & 3& 3 & 3\\ \hline
  5 & 3 & 4& 3& 3\\ \hline
  6-7 & 4 & 4 & 4 & 3\\ \hline
  8-9 & 4 & 4 & 4 & 4\\ \hline
  10 & 5 & 4 & 4 & 4\\ \hline
  11-15 & 5 & 5 & 5 & 4\\ \hline
  16-18 & 5 & 5 & 5 & 5\\ \hline
  19-21 & 5 & 6 & 5 & 5\\ \hline
  22-31 & 6 & 6 & 6 & 5\\ \hline
  32-37 & 6 & 6 & 6 & 6\\ \hline
  38-42 & 7 & 6 & 6 & 6\\ \hline
  43-63 & 7 & 7 & 7 & 6\\ \hline
  64-75 & 7 & 7 & 7 & 7\\ \hline
  76-85 & 7 & 8 & 7 & 7\\ \hline
  86-127 & 8 & 8 & 8 & 7\\ \hline
\end{tabular}
\caption{Values of $q^*_\forall(\ell), q^*_\exists(\ell), q^*(\ell)$ and $r(\ell)$ for $1\leq \ell \leq 127.$} 
\label{q*-table}
\end{center}
\end{table}

We now state two simple corollaries to Lemma \ref{lem:q*-vals}.

\begin{restatable}{corollary}{qinductionseparate}\label{cor:q*inductionseparate} 
We have $q^*_\exists(\ell ) = 2 + q^*_\exists(\lfloor(\ell +1)/4\rfloor)$ for all $\ell \geq 5$. Similarly, we have $q^*_\forall(\ell ) = 2 + q^*_\forall(\lfloor(\ell  + 2)/4\rfloor)$ for all $\ell  \geq 3$.
\end{restatable}

\begin{proof}
It follows directly from Lemma \ref{lem:q*-vals} that $q^*_\exists(\ell) = q^*_\forall(\lceil\ell/2\rceil) + 1$ for $\ell \geq 2$. Similarly, $q^*_\forall(\ell) = q^*_\exists(\lfloor\ell/2\rfloor) + 1$ for $\ell \geq 3$. It follows that for $\ell \geq 1$:
    \begin{equation}
        q^*_\exists(4\ell + 1) = q^*_\exists(4\ell + 2) = q^*_\forall(2\ell + 1) + 1 = q^*_\exists(\ell) + 2. \label{enum1}
    \end{equation}
Similarly, for $\ell \geq 2$:
    \begin{equation}
        q^*_\exists(4\ell - 1) = q^*_\exists(4\ell) = q^*_\forall(2\ell) + 1 = q^*_\exists(\ell) + 2. \label{enum2}
    \end{equation}
Combining these gives us the result (and the associated ranges).
\end{proof}

\begin{restatable}{corollary}{qpowersoftwo}\label{cor:q*-powers-of-two}
For all $k \geq 1$, we have $q^*_\forall(2^k) = q^*_\forall(2^{k-1}) + 1$. Similarly, for all $k \geq 2$, we have $q^*_\exists(2^k) = q^*_\exists(2^{k-1}) + 1$.
\end{restatable}

\begin{proof}
We prove both statements by induction on $k$. For the first, we have $q^*_\forall(1) = 1$ and $q^*_\forall(2) = 2$, establishing the base case. Inductively, by Corollary \ref{cor:q*inductionseparate}, we have:
\begin{equation*}
    q^*_\forall(2^k) = q^*_\forall\left(\left\lfloor\frac{2^k  + 2}{4}\right\rfloor\right) + 2 = q^*_\forall(2^{k-2}) + 2 = q^*_\forall(2^{k-1}) + 1,
\end{equation*}
with the last equality following from the induction hypothesis. Similarly, for the second part, we start with $q^*_\exists(2) = 2$ and $q^*_\exists(4) = 3$, establishing the base case. Inductively, by Corollary \ref{cor:q*inductionseparate}, we have:
\begin{equation*}
    q^*_\exists(2^k) = q^*_\exists\left(\left\lfloor\frac{2^k  + 1}{4}\right\rfloor\right) + 2 = q^*_\exists(2^{k-2}) + 2 = q^*_\exists(2^{k-1}) + 1.\qedhere
\end{equation*}
\end{proof}

\begin{remark}\label{rem:recursive++}
Lemma \ref{lem:q*-vals} and Corollary \ref{cor:q*inductionseparate} are more than just recursive expressions for $q^*_\exists(\ell)$ and $q^*_\forall(\ell)$. By virtue of Lemma \ref{lem:cma-is-well-defined}, we can now read off a quantifier signature establishing $q^*_\exists(2\ell)$ in terms of $q^*_\forall(\ell)$, and analogously for the other expressions.
\end{remark}

We now state and prove the main result of this section.

\begin{restatable}{theorem}{mainlos}
For all $\ell \geq 1$, we have:
\begin{equation*}
    r(\ell) \leq q(\ell) \leq r(\ell) + 1.
\end{equation*}
\end{restatable}

\begin{proof}
The first inequality, $q(\ell) \geq r(\ell)$, is obvious. For the second, we will show that $q^*_\exists(\ell)$ and $q^*_\forall(\ell)$ are both bounded above by $r(\ell) + 1$ (and hence bound $q(\ell)$). From Lemma \ref{lem:q*-vals}:
\begin{align*}
    r(1) = q^*_\forall(1) = q^*_\exists(1) - 1 = 1 && r(2) = q^*_\forall(2) = q^*_\exists(2) = 2,
\end{align*}
so the assertion is true for $\ell \leq 2$. Now, by Corollary \ref{cor:q*-powers-of-two}, we know that $q^*_\forall(2^k) = q^*_\exists(2^k) = k+1$ for $k \geq 1$. By Theorem \ref{rosey-fact}, we also know that $r(2^k) = k+1$ for $k \geq 1$. So the three functions, $r(\cdot)$, $q^*_\forall(\cdot)$, and $q^*_\exists(\cdot)$, all equal each other at successive powers of two, and increase by one between these successive powers. Since all three functions are monotonic, they differ from one another by at most one. Therefore, we have $q^*_\exists(\ell) \leq r(\ell) + 1$ and $q^*_\forall(\ell) \leq r(\ell) + 1$.
\end{proof}

\subsection{Quantifier Alternation Pattern}\label{sec:quantalternation}

In this section, we present two results that will be useful for a few results in Section \ref{sec:strings}.

\begin{restatable}[Alternation Theorem, Smaller vs.~Larger]{theorem}{alternationone}\label{thm:alternation1} 
For every $\ell \geq 1$, there is a separating sentence $\sigma_\ell$ for $(L_{\leq \ell}, L_{> \ell})$ with $q^*(\ell)$ quantifiers, such that the quantifier signature of $\sigma_\ell$ strictly alternates and ends with a $\forall$.
\end{restatable}

\begin{proof}
The theorem is certainly true for small values of $\ell$; e.g., when $\ell = 1$, $q^*(1) = q^*_\forall(1)$, and the sentence corresponding to that strategy has quantifier signature $\forall$. Similarly, when $\ell = 2$, $q^*(1) = q^*_\exists(1)$, and the sentence corresponding to that strategy has quantifier signature $\exists\forall$. The theorem can be verified for $\ell \leq 5$ simply by referring to Table \ref{q*-table}. We now proceed by induction.

Suppose $\ell$ is even, say $\ell = 2\ell'$. There are three cases:
\begin{itemize}
    \item If $q^*_\exists(\ell) < q^*_\forall(\ell)$, this means by Lemma \ref{lem:q*-vals} that $q^*_\forall(\ell') < q^*_\exists(\ell')$. So, by induction, there is a separating sentence $\sigma_{\ell'}$ with quantifier signature $\forall\exists\ldots\forall$ consisting of $q^*_\forall(\ell')$ alternating quantifiers. But by Lemma \ref{lem:q*-vals} and Remark \ref{rem:recursive++}, we can obtain a separating sentence $\sigma_\ell$ with quantifier signature $\exists\forall\ldots\forall$ consisting of $q^*_\exists(\ell)$ alternating quantifiers.
    \item If $q^*_\forall(\ell) < q^*_\exists(\ell)$, this means by Lemma \ref{lem:q*-vals} that $q^*_\exists(\ell') < q^*_\forall(\ell')$. Again, by induction, there is a separating sentence $\sigma_{\ell'}$ with quantifier signature $\exists\forall\ldots\forall$ consisting of $q^*_\exists(\ell')$ alternating quantifiers. By Lemma \ref{lem:q*-vals}, we can obtain a separating sentence $\sigma_\ell$ with quantifier signature $\forall\exists\ldots\forall$ consisting of $q^*_\forall(\ell)$ alternating quantifiers.
    \item If $q^*_\forall(\ell) = q^*_\exists(\ell)$, this means by Lemma \ref{lem:q*-vals} that $q^*_\exists(\ell') = q^*_\forall(\ell')$. Again, by induction, there is a separating sentence $\sigma_{\ell'}$ consisting of $q^*(\ell')$ alternating quantifiers ending with a $\forall$. By Lemma \ref{lem:q*-vals}, we can obtain a separating sentence $\sigma_\ell$ by prepending a quantifier to $\sigma_{\ell'}$ maintaining alternation. This would still contain $q^*(\ell)$ alternating quantifiers.
\end{itemize}

Now suppose $\ell$ is odd, say $\ell = 2\ell' + 1$. There are three cases:
\begin{itemize}
    \item If $q^*_\exists(\ell) < q^*_\forall(\ell)$, this means by Lemma \ref{lem:q*-vals} that $q^*_\forall(\ell' + 1) < q^*_\exists(\ell') \leq q^*_\exists(\ell' + 1)$. By induction, there is a separating sentence $\sigma_{\ell' + 1}$ with quantifier signature $\forall\exists\ldots\forall$ consisting of $q^*_\forall(\ell' + 1)$ alternating quantifiers. Then we can prepend a $\exists$ to obtain a separating sentence $\sigma_\ell$ with quantifier signature consisting of $q^*_\exists(\ell)$ alternating quantifiers.
    \item If $q^*_\forall(\ell) < q^*_\exists(\ell)$, this means by Lemma \ref{lem:q*-vals} that $q^*_\exists(\ell') < q^*_\forall(\ell' + 1)$. If $q^*_\exists(\ell') < q^*_\forall(\ell')$, we are again done by induction. If $q^*_\exists(\ell') = q^*_\forall(\ell')$, this means $q^*_\forall(\ell' + 1) > q^*_\forall(\ell')$, implying by Corollary \ref{cor:q*inductionseparate} that $\ell' \equiv 1\pmod 4$. But then $q^*_\exists(\ell' + 1) = q^*_\exists(\ell') < q^*_\forall(\ell' + 1)$. Therefore, $q^*(\ell') = q^*(\ell' + 1) = q^*_\exists(\ell' + 1)$, and so any alternating quantifier signature with $q^*(\ell')$ quantifiers ending with a $\forall$ must start with a $\exists$. Since by induction, $\sigma_{\ell'}$ has $q^*(\ell')$ alternating quantifiers ending with a $\forall$, it must also start with a $\exists$. Now again, we are done by prepending a $\forall$, by induction.
    \item If $q^*_\forall(\ell) = q^*_\exists(\ell)$, this means by Lemma \ref{lem:q*-vals} that $q^*_\exists(\ell') = q^*_\forall(\ell' + 1)$. Again, by induction, if $q^*_\forall(\ell' + 1) < q^*_\exists(\ell' + 1)$, we are done. If $q^*_\forall(\ell' + 1) = q^*_\exists(\ell' + 1)$, however, we have to be a little more careful. In that situation, if $q^*_\exists(\ell') = q^*_\forall(\ell')$, then $q^*(\ell') = q^*(\ell' + 1)$, and then the sentences $\sigma_{\ell'}$ and $\sigma_{\ell' + 1}$ have the same quantifier signature. Now, depending on the leading quantifier in that signature, we can inductively use either $q^*_\forall(\ell)$ or $q^*_\exists(\ell)$. Otherwise, $q^*_\exists(\ell') > q^*_\forall(\ell')$. But then, $\sigma_{\ell'}$ starts with a $\forall$, and therefore, the sentence $\sigma'_{\ell'}$ that is used by $\bS$ in the $q^*_\exists(\ell')$ strategy has $q^*(\ell') + 1 = q^*(\ell' + 1)$ quantifiers and starts with an $\exists$. Now, by induction, we can obtain a sentence $\sigma_\ell$ using $q^*_\forall(\ell)$ that calls $\sigma'_{\ell'}$, and has $1 + q^*(\ell' + 1)$ alternating quantifiers ending with a $\forall$. Since $q^*(\ell' + 1) = q^*_\forall(\ell' + 1) = q^*_\exists(\ell') = q^*(\ell) - 1$, we are done.
\end{itemize}
This concludes the proof.
\end{proof}

\begin{restatable}[Alternation Theorem, One vs.~All]{theorem}{alternationtwo}\label{thm:alternation2} 
For every $\ell \geq 1$, there is a sentence $\varphi_\ell$ separating $L_\ell$ from all other linear orders. This sentence $\varphi_\ell$ has an alternating quantifier signature (ending with a $\forall$) consisting of at most $q^*(\ell) + 2$ quantifiers.
\end{restatable}

\begin{proof}
When $\ell = 1$, the theorem follows directly from Theorem \ref{thm:alternation1}. So suppose $\ell > 1$.

Again let $\cA = \{L_\ell\}$, and let $\cB = \cB_1 \sqcup \cB_2$, where $\cB_1 = L_{\leq \ell - 1}$, and $\cB_2 = L_{> \ell}$. By Theorem \ref{thm:alternation1}, there is a sentence $\sigma_{\leq \ell}$ that is true for $L_{\leq \ell}$ and false for $L_{> \ell} = \cB_2$, with the given alternating quantifier signature, with $q^*(\ell)$ quantifiers. Similarly, there is a sentence $\sigma_{\leq \ell - 1}$ which is true for $L_{\leq \ell - 1} = \cB_1$ and false for $L_{\geq \ell}$, with the given alternating quantifier signature, with $q^*(\ell - 1) \leq q^*(\ell)$ quantifiers. Assume these two sentences both have $q^*(\ell)$ quantifiers (possibly by prepending a dummy leading quantifier to $\sigma_{\leq \ell - 1}$). Let $\sigma_2 := \sigma_{\leq \ell}$ and $\sigma_1 := \lnot\sigma_{\leq \ell - 1}$. Note that $\sigma_1$ separates $(\cA, \cB_1)$ (say with strategy $\cS_1$), and $\sigma_2$ separates $(\cA, \cB_2)$ (say with strategy $\cS_2$), and so $\sigma_1 \land \sigma_2$ separates $(\cA, \cB)$. Furthermore, $\sigma_1$ and $\sigma_2$ both have alternating quantifier signatures of the same length $q^*(\ell)$, but they are complements of each other: $\sigma_2$ ends in a $\forall$, and $\sigma_1$ ends in a $\exists$.

Consider the MS game on $(\cA, \cB)$. We will now give a $\bS$ strategy. $\bS$ always starts with a universal move. Exactly one of the sentences $\sigma_1$ and $\sigma_2$ begins with a $\forall$.

If the sentence with a leading $\forall$ is $\sigma_1$, $\bS$ plays his round $1$ moves, playing pebble $\r$ on the element $\mathsf{max}$ in all boards in $\cB_2$, and according to the strategy $\cS_1$ in all boards in $\cB_1$. Note that, by virtue of $\cS_1$ being the $\mathsf{CMA}$ strategy, $\bS$ never plays $\r$ on the element $\mathsf{max}$ in any board in $\cB_1$. Therefore, every board in $\cB_1$ satisfies $\r \neq \mathsf{max}$, whereas every board in $\cB_2$ satisfies $\r = \mathsf{max}$. Once $\bD$ has responded obliviously, partition $\cA$ as $\cA_1 \sqcup \cA_2$ such that every board in $\cA_1$ satisfies $\r \neq \mathsf{max}$, whereas every board in $\cA_2$ satisfies $\r = \mathsf{max}$ as well.

Now, the sub-games $(\cA_1, \cB_1)$ and $(\cA_2, \cB_2)$ can be played in parallel; there will be no matching pair from $\cA_i$ and $\cB_j$ for $i \neq j$; furthermore, the two strategies both have patterns that are subsequences of the sequence $(\forall, \exists, \ldots, \exists, \forall)$, which has length $q^*(\ell) + 1$ or $q^*(\ell) + 2$ depending on the parity of $q^*(\ell)$. Therefore, by Lemma \ref{lem:parallelplay}, the result follows.
\end{proof}
\section{Strings}\label{sec:strings}

In this section, we consider binary strings and the number of quantifiers needed to separate two sets of strings. We would like to bound the number of quantifiers required for these separations as a function of both the length of the strings, as well as the sizes of the sets.

\subsection{Separating One String from Another String}\label{sec:one-vs-one}

\begin{restatable}[One vs.~One, $n$-bit]{proposition}{onevsonen}\label{prop:one-vs-one-n}
    \textbf{Upper Bound:} For every pair $w, w'$ of $n$-bit strings such that $w \neq w'$, there is a sentence $\varphi_{w, w'}$ with $\log(n) + O(1)$ quantifiers separating $(\{w\}, \{w'\})$. When written in prenex normal form, this sentence $\varphi_{w, w'}$ has an alternating quantifier signature ending with $\forall$.\\
    \textbf{Lower Bound:} For all sufficiently large $n$, there exist two $n$-bit strings $w, w'$, such that separating them requires $\lfloor\log(n)\rfloor$ quantifiers. 
\end{restatable}

\begin{proof}
    \textbf{Upper Bound:} Let $w, w' \in \{0,1\}^n$ be any two distinct strings. There is an index $i \in [n]$ such that $w_i \neq w'_i$. Let $\cA = \{w\}$ and $\cB = \{w'\}$. We will show that $\bS$ wins the MS game on $(\cA, \cB)$ in $\log(n) + O(1)$ moves.
    
    In round $1$, $\bS$ plays his move with pebble $\r$ on the $\cA$ side, on the element $w_i$ in $w$, creating the pebbled string $\langle w ~|~ w_i\rangle$. Assume $\bD$ responds obliviously on the $\cB$ side. We can now immediately use Observation \ref{obs:discard} to discard the resulting pebbled string $\langle w' ~|~ w'_i \rangle \in \cB$, where the pebble $\r$ is on the element $w'_i$.
    
    Now, every remaining board in $\cB$ is of the form $\langle w' ~|~ w'_j\rangle$, for $j \neq i$. Note that the substring $w'[1, j]$ has length $j$, which is different from $i$, the length of the substring $w[1, i]$ of $w \in \cA$. So from this point on, $\bS$ can simply play the strategy he uses in the proof of Theorem \ref{thm:alternation2} to separate a linear order of length $i$ from all other linear orders, which he wins in $q^*(i) + 3 = \log(i) + O(1) \leq \log(n) + O(1)$ rounds with an alternating pattern. Recall that one consequence of adopting this strategy is that $\bS$ always plays in every string between $\min$ and $\r$ (both inclusive); any of $\bD$'s responses outside that range can be discarded.
    
    The number of quantifiers required is $\log(n) + O(1)$. The alternating quantifier signature follows from the fact that the strategy in Theorem \ref{thm:alternation2} started with a universal move.
    
    \textbf{Lower Bound:} Let $\ell = 2^k + 2$ for $k > 1$. We construct a pair of $\ell$-bit strings that are hard to distinguish. Let $w = 0^{2^{k-1}} 100^{2^{k-1}}$ and $w' = 0^{2^{k-1}} 010^{2^{k-1}}$. 
    If $\bS$ plays entirely on one side of the respective $1$s then he is effectively playing an MS game $(L_{2^{k-1}}, L_{2^{k-1}-1})$. By Theorem \ref{rosey-fact} and our assumption that $k > 1$, we have $r(2^{k-1}) = k = \lfloor \log(\ell) \rfloor$. Since the quantifier rank lower bounds the number of quantifiers, the MS game played in this fashion requires at least $\lfloor \log(\ell)\rfloor$ rounds to win.

    Now suppose that instead of playing entirely on the same side of the respective $1$s, $\bS$ plays on both sides of a $1$ and/or on the $1$ during these $\lfloor \log(\ell)\rfloor$ rounds. In this case, $\bD$ can play obliviously to the left of the 1 when $\bS$ plays to the left of the 1, obliviously to the right of the 1 when $\bS$ plays to the right of the 1, and on the $1$ whenever $\bS$ plays on the $1$, thereby keeping matching pairs simultaneously on both sides. The lower bound follows.
\end{proof}

Of course, for the upper bound in Proposition \ref{prop:one-vs-one-n}, if $w'$ is not an $n$-bit string to begin with, then $\bS$ can just exhibit the length difference between them directly using Theorem \ref{thm:alternation2}. This immediately gives us the following corollary.

\begin{restatable}[One vs.~One, arbitrary length]{corollary}{onevsone}\label{cor:one-vs-one}
Each $n$-bit string $w$ can be separated from any other string $w'$ with a sentence that has $\log(n) + O(1)$ quantifiers, with an alternating quantifier signature ending with $\forall$.
\end{restatable}

The results above give tight bounds for the number of quantifiers needed to separate one arbitrary string from any other. In the next subsection, we generalize this result.

\subsection{Separating One String from Many Other Strings}\label{sec:one-vs-many}

In this section, we will need a slight generalization of Lemma \ref{lem:parallelplay}, enabling us to partition a game into more than two sub-games when necessary. We state this generalization here without a proof, since it follows immediately from Lemma \ref{lem:parallelplay}.

\begin{lemma}[Generalized Parallel Play Lemma]\label{lem:genparallelplay}
Let  $\cA$ and $\cB$ be two sets of pebbled structures, and let $r \in \mathbb{N}$. Let $P \in \{\exists, \forall\}^r$ be a sequence of quantifiers of length $r$. Suppose that $\cA$ and $\cB$ can be partitioned as $\cA = \cA_1 \sqcup \ldots \sqcup \cA_k$ and $\cB = \cB_1 \sqcup \ldots \sqcup \cB_k$ respectively, such that for all $1 \leq i \leq k$, $\bS$ has a winning strategy $\cS_i$ for the $r_i$-round MS game on $(\cA_i,\cB_i)$ (where $r_i \leq r$), satisfying:
\begin{enumerate}
\item For all $i$, $\mathsf{pat}(\cS_i)$ is a subsequence of $P$.
\item At the end of the sub-games, for $i \neq j$, there does not exist a board in $\cA_i$ and a board in $\cB_j$ forming a matching pair.
\end{enumerate}
Then $\bS$ wins the $r$-round MS game on $(\cA,\cB)$ with pattern $P$.
\end{lemma}

We are now ready to prove our main results in this section. We start by separating a single string from \emph{all} others of the same length.

\begin{restatable}[One vs.~All, $n$-bit]{theorem}{onevsalln}\label{thm:one-vs-all-n}
Every $n$-bit string can be separated from the set of all other $n$-bit strings by a sentence with $3\lceil \log_3(n)\rceil$ quantifiers. This sentence has quantifier signature consisting of $\lceil\log_3(n)\rceil$ iterations of $\exists\exists\forall$.
\end{restatable}

\begin{proof}
Assume first that $n$ is a power of $3$. Let $w \in \{0, 1\}^n$ be arbitrary. We will demonstrate an inductive strategy for $\bS$ to win the MS game on $(\cA, \cB)$, where $\cA = \{w\}$, and $\cB = \{0, 1\}^n - \{w\}$, 
in $3\log_3(n)$ rounds. To do this, we will inductively separate each string of length $n$ from each other string of length at most $n$.

$\bS$ starts by breaking $w$ into thirds. Let $i = n/3$, and $j = 2n/3$. $\bS$ plays his first two moves on the $\cA$ side, placing the pebble $\r$ on $w_{i+1}$, and the pebble $\b$ on $w_{j + 1}$, creating the pebbled string $\langle w ~|~ w_{i+1}, w_{j+1} \rangle \in \cA$. Note that the three substrings $w[1, i]$, $w[i+1, j]$, and $w[j+1, n]$, all have length $n/3$. We may assume that $\bD$ responds obliviously. We may also further assume that we then discard all boards on the $\cB$ side that are not partially isomorphic to the board in $\cA$.

Now, every pebbled string in $\cB$ is of the form $\langle w' ~|~ w'_{i' + 1}, w'_{j' + 1}\rangle$, with $0 \leq i' \leq j' < n$ corresponding to indices such that $\r$ is on $w'_{i' + 1}$ and $\b$ is on $w'_{j' + 1}$. We now claim that, for each such pebbled string:
\begin{itemize}
    \item At least one of the substrings $w'[1, i']$, $w'[i' + 1, j']$, and $w'[j' + 1, n]$ has length $n/3$ or less; otherwise the entire string $w'$ has length more than $n$.
    \item If all three substrings above have length exactly $n/3$, then one of them is \emph{different} from the corresponding substring of $w$; otherwise, we would have $w = w'$.
\end{itemize}

We now partition $\cB = \cB_1 \sqcup \cB_2 \sqcup \cB_3$ as follows:
\begin{itemize}
    \item for each $\langle w' ~|~  w'_{i' + 1}, w'_{j' + 1} \rangle \in \cB_1$, either $w'[1, i']$ has strictly fewer than $n/3$ bits, or it has exactly $n/3$ bits but is distinct from $w[1, i]$.
    \item for each $\langle w' ~|~  w'_{i' + 1}, w'_{j' + 1} \rangle \in \cB_2$, either $w'[i' + 1, j']$ has strictly fewer than $n/3$ bits, or it has exactly $n/3$ bits but is distinct from $w[i+1, j]$.
    \item for each $\langle w' ~|~  w'_{i' + 1}, w'_{j' + 1} \rangle \in \cB_3$, either $w'[j'+1, n]$ has strictly fewer than $n/3$ bits, or it has exactly $n/3$ bits but is distinct from $w[j+1, n]$.
\end{itemize}

In the next round, $\bS$ places the pebble $\g$ on the $\cB$ side: for every $\langle w' ~|~  w'_{i' + 1}, w'_{j' + 1} \rangle \in \cB_1$, he places $\g$ on $w'_1$ (i.e., on the constant $\mathsf{min}$); for every $\langle w' ~|~  w'_{i' + 1}, w'_{j' + 1} \rangle \in \cB_2$, he places $\g$ on $w'_{i'+1}$ (on top of $\r$); and for every $\langle w' ~|~  w'_{i' + 1}, w'_{j' + 1} \rangle \in \cB_3$, he places $\g$ on $w'_{j'+1}$ (on top of $\b$). In all cases, he is pointing out the leftmost element of the ``mismatched'' portion of the pebbled string from $\langle w ~|~  w_{i + 1}, w_{j + 1} \rangle$.

In response, $\bD$ plays obliviously on the $\cA$ side. We may once again assume that we discard all resulting boards that are not partially isomorphic to any of the boards in $\cB$. This means there are only (at most) three boards remaining in $\cA$:
\begin{itemize}
    \item let $\cA_1$ consist of the pebbled string $\langle w ~|~  w_{i + 1}, w_{j + 1}, w_1 \rangle$, where $\r$, $\b$, and $\g$ are on $w_{i+1}$, $w_{j+1}$, and $w_1$ respectively.
    \item let $\cA_2$ consist of the pebbled string $\langle w ~|~  w_{i + 1}, w_{j + 1}, w_{i+1} \rangle$, where $\r$, $\b$, and $\g$ are on $w_{i+1}$, $w_{j+1}$, and $w_{i+1}$ respectively.
    \item let $\cA_3$ consist of the pebbled string $\langle w ~|~  w_{i + 1}, w_{j + 1}, w_{j+1} \rangle$, where $\r$, $\b$, and $\g$ are on $w_{i+1}$, $w_{j+1}$, and $w_{j+1}$ respectively.
\end{itemize}

Note that, if $k \neq \ell$, there can never be a board from $\cA_k$ and a board from $\cB_\ell$ forming a matching pair.

We can now inductively play three parallel \textbf{one-vs-all} games, on $(\cA_1, \cB_1)$, $(\cA_2, \cB_2)$, and $(\cA_3, \cB_3)$, using Lemma \ref{lem:genparallelplay}. In each case, $\bS$ proceeds recursively, always playing two existential moves in the ``correct'' third on the $\cA$ side, and then using a universal move to point out the left endpoint of the mismatched third on the strings on the $\cB$ side. In all cases, we always discard any of $\bD$'s moves that are outside the relevant third of the strings. The prenex quantifier signature established by this game is
$\exists\exists\forall\cdots\exists\exists\forall$. This process terminates once we have discarded all strings on the $\cB$ side, at which point, all isomorphisms have been broken.

Note that $\bS$ always decreases the size of the string by a factor of $3$ in each successive sub-game, after spending two existential moves followed by a universal one, which corresponds to the block $\exists\exists\forall$. Therefore, the process continues for $\log_3(n)$ iterations, and so, the number of quantifiers used in total is $3\log_3(n)$.

The procedure is similar when $n$ is not a power of $3$. $\bS$ can play his existential moves on $w_{\lfloor n/3\rfloor}$ and $w_{\lfloor 2n/3\rfloor}$, and point out the difference on the right in exactly the same way. The same argument holds, and the sub-games now have size at most $n/3$. Therefore, the process still takes $\lceil\log_3(n)\rceil$ iterations and the total number of quantifiers used is still $3\lceil\log_3(n)\rceil$.
\end{proof}

Note that the method described in the proof of Theorem \ref{thm:one-vs-all-n} can be adapted easily to yield $\lceil k\log_k(n)\rceil$ quantifiers for any $k \geq 2$, using the same technique. This quantity is simply minimized at $k = 3$. As a corollary, we have the following.

\begin{restatable}[One vs.~All, arbitrary lengths]{corollary}{onevsall} \label{cor:one-vs-all}
Every $n$-bit string can be separated from the set of all other strings (of any length) by a sentence with $3\log_3(n) + O(1)$ quantifiers. 
\end{restatable}

\begin{proof}
    Let $w \in \{0, 1\}^n$ be arbitrary. Let $\cA = \{w\}$, and let $\cB \subseteq \{0, 1\}^\ast - \cA$. Assume WLOG that $n > 1$, and that $\cB$ does not contain any one-bit string (if it does, those strings can be identified in $O(1)$ rounds).
    
    $\bS$ starts by playing a universal move on the $\cB$ side, where he plays on $w'_{\mathsf{min}}$ if $|w'| \leq n$, and otherwise plays on $w'_{\mathsf{max}}$. This partitions $\cB$ into $\cB_1 \sqcup \cB_2$, where $\cB_1 = \{\langle w' ~|~ w'_{\mathsf{min}}\rangle : |w'| \leq n\}$ and $\cB_2 = \{\langle w' ~|~ w'_{\mathsf{max}}\rangle : |w'| > n\}$, where the pebbled strings in $\cB_1$ and $\cB_2$ are in different isomorphism classes (note that $\mathsf{min}$ and $\mathsf{max}$ correspond to distinct elements). Once $\bD$ responds, we can again discard all boards in $\cA$ that are not isomorphic to either of the two isomorphism classes in the $\cB$ side. This partitions $\cA$ into $\cA_1 \sqcup \cA_2$ as well, where $\cA_1 = \{\langle w ~|~ w'_{\mathsf{min}}\rangle\}$, and $\cA_2 = \{\langle w ~|~ w'_{\mathsf{max}}\rangle\}$. From this point, the method of Theorem \ref{thm:one-vs-all-n} can be immediately adopted by $\bS$ to win the game $(\cA_1, \cB_1)$ using $3\lceil\log_3n\rceil$ quantifiers.
    
    Now consider the sub-game $(\cA_2, \cB_2)$. Note that $\bS$ can adopt the exact same strategy as above, except to use his universal moves to mark the left endpoints of a substring if it is too big rather than too small (or a mismatch if all three have equal lengths). The process still iterates with the same pattern, converging in $\lceil\log_3(n)\rceil$ iterations. After the final step, once all boards not matching the isomorphism classes on the other side are discarded, $\bS$ is left with possibly some boards on the $\cB$ side where there are unpebbled elements, that cannot be matched in the board on the $\cA$ side with that corresponding isomorphism class. So for his last move, $\bS$ just plays a universal move on each of these unpebbled elements, thereby breaking all induced isomorphisms. Since the pattern in one sub-game is a subsequence of the pattern in the other one, we can apply Lemma \ref{lem:genparallelplay}, and the result follows. Note that the $O(1)$ term absorbs the first and last plays, as well as any additive factors from the ceiling.
\end{proof}

With this baseline established, we can get improved results when we are separating an $n$-bit string from relatively fewer other $n$-bit strings. We parameterize the number of these other strings as a function $f(n)$ of the length of the strings.

Our first result, Proposition \ref{prop:one-vs-f(n)-baby}, separates an $n$-bit string from $f(n)$ other $n$-bit strings using $\log(n) + \log(f(n)) + O(1)$ quantifiers. Note that when $f(n) = \Omega(n)$, this proposition only gives us $2\log(n) + O(1)$ quantifiers, whereas we know from Theorem \ref{thm:one-vs-all-n} that we can already solve this problem with just over $3\log_3(n) \approx 1.89\log(n)$ quantifiers. So, Proposition \ref{prop:one-vs-f(n)-baby} is only interesting when $f(n) = o(n)$. We encourage the reader to study the full proof, since it will significantly elucidate the proof of Theorem \ref{thm:one-vs-f(n)}.

\begin{restatable}[One vs. Many, $n$-bit --- \emph{Basic}]{proposition}{onevmanynbit}\label{prop:one-vs-f(n)-baby}
Let $f : \mathbb{N} \to \mathbb{N}$ be a function. Every $n$-bit string can be separated from any set of $f(n)$ other $n$-bit strings by a sentence with $\log(n) + \log(f(n)) + O(1)$ quantifiers.
\end{restatable}

\begin{proof}
Let $w$ be any $n$-bit string. Let $\cA = \{w\}$, and $\cB \subseteq \{0, 1\}^n - \cA$ satisfying $|\cB| = f(n)$. Note that $f(n) < 2^n$.

The key to this proof is that $\bS$ plays an initial set of ``preprocessing'' universal moves. For this purpose, let $\nabla : \cB \to \{0, 1\}^{\lceil\log (f(n))\rceil} - \{1^n\}$ be a one-to-one function that maps the strings in $\cB$ into a set of ``instructional strings,'' which are $\lceil\log (f(n))\rceil$-bit binary strings. We disallow $1^n$ from the range.  Note that the size of the range is $2^{\lceil\log (f(n))\rceil}$ unless $\lceil\log (f(n))\rceil = n$, in which case, since we exclude $1^n$, the size of the range is $2^n - 1$. 
In either case, the size of the range is at least the size of the domain. It follows that our needed 1-1 function $\nabla$ exists. As we shall see, for any $w' \in \cB$, the string $\nabla(w')$ encodes how $\bS$ plays on the string $w'$.

$\bS$ makes his first $\lceil \log(f(n)) \rceil + 1$ moves in $\cB$. Consider any $w' \in \cB$. In round $1$, $\bS$ plays pebble $1$ on $w'_1$. For $1 \leq i \leq \lceil \log(f(n)) \rceil$, suppose pebble $i$ is on $w'_j$. In round $(i + 1)$, $\bS$ looks at the $i$th bit of $\nabla(w')$, and plays on $w'_j$ if that bit is $0$, or on $w'_{j+1}$ if that bit is $1$.

After $\lceil \log(f(n))\rceil + 1$ rounds, no two strings in $\cB$ are in the same isomorphism class. After $\bD$'s oblivious responses, all of these isomorphism classes are represented on the $\cA$ side as well. By partitioning the two sides appropriately, this creates $f(n)$ parallel \textbf{one-vs-one} MS games. $\bS$ can win each of these games using the same alternating pattern of length $\log n + O(1)$, using the strategy described in Proposition \ref{prop:one-vs-one-n}. The result now follows from Lemma \ref{lem:genparallelplay}.
\end{proof}

By a more clever argument, we can dramatically improve Proposition \ref{prop:one-vs-f(n)-baby}:

\begin{restatable}[One vs. Many, $n$-bit --- \emph{Improved}]{theorem}{onevsfn}\label{thm:one-vs-f(n)}
Let $t \geq 2$ be any integer, and $f : \mathbb{N} \to \mathbb{N}$ be a function such that $\lim_{n \rightarrow \infty} f(n) = \infty$. Then, for all $n$, it is possible to separate each $n$-bit string from any set of $f(n)$ others by a sentence with $\log(n) + \log_t(f(n)) + O_t(1)$ quantifiers.
\end{restatable}

\begin{proof}
Of course, $f(n) < 2^n$. Since $\lim_{n \rightarrow \infty} f(n) = \infty$, there is an $N$ such that for all $n \geq N$, we have $f(n) \geq t^{et}$. When $n < N$, the theorem is certainly true, as the number of quantifiers can be absorbed into the $O_t(1)$ term. So for the rest of the proof, assume $f(n) \geq t^{et}$.

As in Proposition \ref{prop:one-vs-f(n)-baby}, $\bS$ will first use $m = \lceil \log_t(f(n)) \rceil$ ``preprocessing'' universal moves up front. Note that since $f(n) < 2^n$ and $t \geq 2$, we have $m \leq n$. $\bS$ can permute these $m$ pebbles any way he wants, to create distinct orderings, and therefore, distinct isomorphism types. Using Stirling's approximation, the number of such permutations can be lower bounded as:
\begin{equation*}
m! \geq \sqrt{2\pi m}\left(\frac{m}{e}\right)^m = \omega\left(\left(\frac{m}{e}\right)^m\right) = \omega\left(\left(\frac{\lceil \log_t(t^{et}) \rceil}{e}\right)^m\right) = \omega \big(t^{\lceil \log_t(f(n)) \rceil}\big) = \omega(f(n)).
\end{equation*}

Hence, with $m \leq \log_t(f(n)) + 1$ universal moves, each string in $\cB$ can be associated with its own isomorphism type. The rest of the argument follows the proof of Proposition \ref{prop:one-vs-f(n)-baby}.
\end{proof}


\begin{restatable}[One vs.~Polynomially Many, $n$-bit]{corollary}{onevspolymanyn}\label{cor:1+epsilon-upper}
Let $f : \mathbb{N} \to \mathbb{N}$ be a function satisfying $\lim_{n \rightarrow \infty} f(n) = \infty$, and $f(n) = O(n^k)$ for some constant $k$. Then, for all $n$, and for every $\varepsilon > 0$, it is possible to separate each $n$-bit string from any set of $f(n)$ others by a sentence with $(1 + \varepsilon)\log(n) + O_{k,\varepsilon}(1)$ quantifiers.
\end{restatable}

\begin{proof}
Pick a sufficiently large constant $k > 0$ such that $f(n) \leq  n^k$ for all $n$. Pick $t \geq 2$ large enough so that $k/\log(t) < \varepsilon$. By Theorem \ref{thm:one-vs-f(n)}, denoting by $C(t)$ a constant that depends only on $t$, the number of quantifiers we need is:
\begin{equation*}
\log(n) + \log_t(f(n)) + C(t) \leq \log(n) + k\cdot\log_t(n) + C(t) = \log(n)\cdot\left(1 + \frac{k}{\log(t)}\right) + C(t),
\end{equation*}
which is bounded above by $(1 + \varepsilon)\log(n) + C(t)$. Since $C(t)$ depends only on $t$, which depends only on $k$ and $\varepsilon$, the corollary follows.
\end{proof}

Theorem \ref{thm:one-vs-f(n)} and Corollary \ref{cor:1+epsilon-upper} assumed that $\lim_{n \rightarrow \infty} f(n) = \infty$. The following result, which is just a variation on the same theme, allows us to say something about bounded functions.

\begin{restatable}[One vs.~Boundedly Many]{theorem}{onevsboundedmany}\label{thm:one-vs-constant-best}
Let $f:\mathbb{N} \to \mathbb{N}$ be a function such that $\max_n f(n) = N$. Let $t$ be the largest integer $\geq 2$ satisfying $t^{et} \leq N$. Then, it is possible to separate each $n$-bit string from any set of $f(n)$ others by a sentence with  $\log(n) + \log_tN + O(1)$ quantifiers.
\end{restatable}

\begin{proof}
Using Stirling's approximation, we get $\lceil\log_tN\rceil! \geq f(n)$, similar to the proof of Theorem \ref{thm:one-vs-f(n)}. So once again, $\bS$ can just play $\lceil\log_tN\rceil$ ``preprocessing'' universal moves to separate all $f(n)$ strings in $\cB$ into their own isomorphism classes, and proceed as usual.
\end{proof}

\subsection{Separating Many Strings from Many Other Strings}\label{sec:many-vs-many}

In this section, we generalize the results from Section \ref{sec:one-vs-many} to the case when $\cA$ also starts out with multiple strings. We start by generalizing Theorem \ref{thm:one-vs-f(n)}.

\begin{restatable}[Many vs.~Many, $n$-bit]{theorem}{manyvsmanyn}\label{thm:f(n)-vs-g(n)}
Let $t \geq 2$ be any integer, and let $f, g : \mathbb{N} \to \mathbb{N}$ be functions such that $\lim_{n \rightarrow \infty} f(n) = \infty$ and $\lim_{n \rightarrow \infty} g(n) = \infty$. Then, for all $n$, it is possible to separate each set of $f(n)$ $n$-bit strings from any set of $g(n)$ others by a sentence with $\log(n) + \log_t(f(n)\cdot g(n)) + O_t(1)$ quantifiers.
\end{restatable}

\begin{proof}
Analogously to the proof of Theorem \ref{thm:one-vs-f(n)}, we take $N$ such that for all $n \geq N$, we have $f(n) \geq t^{et}$ and $g(n) \geq t^{et}$. Again, the theorem is true for $n < N$, by absorbing the number of quantifiers into the $O_t(1)$ term. So once again, we may assume that $f(n) \geq t^{et}$ and $g(n) \geq t^{et}$. Once again, $\bS$ can first use up $\lceil\log_t(f(n))\rceil$ ``preprocessing'' existential moves to place each string in $\cA$ in its own isomorphism class. Then he can use up $\lceil\log_t(g(n))\rceil$ ``preprocessing'' universal moves to place each string in $\cB$ in its own isomorphism class. Once this is done, $\bS$ has partitioned the game instance into $\lceil\log_t(f(n))\rceil\cdot\lceil\log_t(g(n))\rceil$ parallel instances of \textbf{one-vs-one} games, each of which has a winning strategy in $\log_2n + O(1)$ rounds, that follows the same alternating pattern as described in Proposition \ref{prop:one-vs-one-n}. Pebbled strings from different parallel instances can never form the same isomorphism class, by construction. Therefore, using arguments from Proposition \ref{prop:one-vs-f(n)-baby} and Lemma \ref{lem:genparallelplay}, we obtain the result. Note that $\log_t(f(n)g(n)) = \log_t(f(n)) + \log_t(g(n))$.
\end{proof}

We can also generalize Corollary \ref{cor:1+epsilon-upper} using Theorem \ref{thm:f(n)-vs-g(n)}. This proof is omitted.

\begin{restatable}[Polynomially Many vs.~Polynomially Many, $n$-bit]{corollary}{polymanyvspolymany}\label{cor:1+epsilon-upper-2-sided}
Let $f, g : \mathbb{N} \to \mathbb{N}$ be functions such that $\lim_{n \rightarrow \infty} f(n), g(n) = \infty$ and $f(n), g(n) = O(n^k)$ for some constant $k$. Then, for all $n$, and for every $\varepsilon > 0$, it is possible to separate each set of $f(n)$ $n$-bit strings from any set of $g(n)$ others by a sentence with $(1 + \varepsilon)\log(n) + O_{k, \varepsilon}(1)$ quantifiers.
\end{restatable}

Finally, we can combine Theorem \ref{thm:one-vs-all-n} and the proof of Theorem \ref{thm:one-vs-f(n)} to obtain a more general upper bound:

\begin{restatable}[Many vs.~All, $n$-bit]{theorem}{manyvsall}\label{thm:f(n)-v-all}
Let $t \geq 2$ be any integer, and let $f : \mathbb{N} \to \mathbb{N}$ satisfy $\lim_{n \rightarrow \infty} f(n) = \infty$. Then, for all $n$, it is possible to separate any set of $f(n)$ $n$-bit strings from all other $n$-bit strings by a sentence with $3\log_3(n) + \log_t(f(n)) + O_t(1)$ quantifiers.
\end{restatable}

\begin{proof}
We may assume WLOG that the $f(n)$ strings are in $\cA$. Using Theorem \ref{thm:one-vs-f(n)}, $\bS$ can expend $\lceil\log_t(f(n))\rceil$ ``preprocessing'' existential moves to put each string in $\cA$ in its own isomorphism class. Once this is done, $\bS$ can play resulting $f(n)$ different parallel \textbf{one-vs-all} games according to Theorem \ref{thm:one-vs-all-n}, which all have the same winning pattern and therefore can be played in parallel using Lemma \ref{lem:genparallelplay}. This expends an additional $3\log_3(n) + O(1)$ quantifiers. The result follows.
\end{proof}

The following corollary follows exactly as Corollary \ref{cor:1+epsilon-upper}, and the proof is omitted.

\begin{restatable}[Polynomially Many vs.~All, $n$-bit]{corollary}{polymanyvsall}\label{cor:3+epsilon-upper}
Let $f : \mathbb{N} \to \mathbb{N}$ be a function satisfying $\lim_{n \rightarrow \infty} f(n) = \infty$ and $f(n) = O(n^k)$ for some constant $k$. Then, for all $n$, and for every $\varepsilon > 0$, it is possible to separate each set of $f(n)$ $n$-bit strings from all other $n$-bit strings by a sentence with $(3 + \varepsilon)\log_3(n) + O_{k, \varepsilon}(1)$ quantifiers.
\end{restatable}

We make an observation at this juncture. In Theorem \ref{thm:one-vs-f(n)}, in the worst case scenario, $f(n) = 2^n - 1$. However, in that case, Theorem \ref{thm:one-vs-f(n)} does not give us any useful bound (and, in fact, Theorem \ref{thm:one-vs-all-n} is significantly better). Similarly, in Theorem \ref{thm:f(n)-vs-g(n)}, if \textit{both} $f(n)$ and $g(n)$ are $\Theta(2^n)$, we do not get anything useful.

In fact, for any arbitrary instance $(\cA, \cB)$ consisting of only $n$-bit strings, $\bS$ always has a winning strategy in $n$ rounds: he can simply use $n$ existential (or universal) moves to play different pebbles on all the elements of each board on one side, which breaks all isomorphisms on the other side. This strategy effectively names the entirety of $\cA$ (or $\cB$). We improve on this na\"{i}ve bound in the following theorem.

\begin{restatable}[Any vs.~Any, $n$-bit --- \emph{Upper Bound}]{theorem}{anyvsanyupper}\label{thm:all-vs-all upper bound}
For all $n$, and for every $\varepsilon > 0$, any two disjoint sets of $n$-bit strings are separable by a sentence with
$(1 + \varepsilon)\frac{n}{\log(n)} + O_\varepsilon(1)$ quantifiers.
\end{restatable}

\begin{proof}
Let $r > 2$ be any real number.

\begin{claim}\label{claim:realr}
For every real $r > 2$, there is some $N_r \in \mathbb{N}$ such that for all $n \geq N_r$, every pair $\cA$ and $\cB$ of disjoint sets of $n$-bit strings are separable by a sentence with $\left\lceil n/\log_r(n)\right\rceil + 3\lceil\log_3(n)\rceil$ quantifiers.
\end{claim}
\begin{proof}[Proof of Claim \ref{claim:realr}]
The idea is once again for $\bS$ to play enough ``preprocessing'' existential moves at the start, to give each string in $\cA$ its own isomorphism class. $\bS$ plays $m := \left\lceil n/\log_r(n)\right\rceil$ such moves. We first note that as long as $n > r^2$, we have $m < \lceil n/2\rceil \leq n$, and so $\bS$ has sufficient space to play the pebbles from these preprocessing moves on distinct elements on each of the boards in $\cA$.

Consider the number of orderings of $m$ pebbles, which is the number of distinct isomorphism classes $\bS$ will create. By Stirling's approximation, we have:
\begin{equation*}
m! = \left\lceil \frac{n}{\log_r(n)} \right\rceil! > \left( \frac{n}{e\log_r(n)} \right)^{\frac{n}{\log_r(n)}}.
\end{equation*}
We wish to show that the right hand side of this equation is at least $2^n$ (for large enough $n$) to account for the largest possible size of $\cA$. Therefore, taking base-$2$ logarithms, we wish to show that:
\begin{equation*} 
\frac{n}{\log_r(n)}\cdot\left(\log(n) - \log(e\log_r(n)) \right) \geq n, \text{~~ i.e., ~~}\log(n) - \log(e) - \log\log_r(n) \geq \log_r(n).
\end{equation*}
Equivalently, we need to show that:
\begin{equation*} 
    \log(n) \geq \frac{\log(n)}{\log(r)} + \log\log_r(n)  + \log(e),
\end{equation*}
or in other words:
\begin{equation*}
\log(n)\left(1 - \frac{1}{\log(r)}\right) \geq \log\log(n) - \log\log(r) + \log(e).
\end{equation*}
Because $r > 2$, we have $\log(r) > 1$, and so the left hand side above grows linearly in $\log(n)$, whereas the right hand side grows logarithmically in $\log(n)$. Hence, there is some integer $N'_r$ such that for all $n \geq N'_r$, the left hand side dominates, and so, indeed, $\left\lceil n/\log_r(n)\right\rceil$ preprocessing moves are enough for each string to get its own isomorphism class. 

Once these preprocessing moves have been played, and $\bD$ has responded obliviously in $\cB$, we can use Observation \ref{obs:discard} to discard all boards in $\cB$ that do not form a matching pair with any board in $\cA$. We are now left with $|\cA|$ isomorphism classes on both sides. We can therefore partition $\cB$ as $\cB_1 \sqcup \ldots \sqcup \cB_{|\cA|}$ according to those isomorphism classes. Now observe that we have partitioned the game into $|\cA|$ parallel instances of \textbf{one-vs-many} games, $(\cA_i, \cB_i)$ for $1 \leq i \leq |\cA|$, where $\cA_i$ is a singleton for all $i$.

We are now in a position to apply Lemma \ref{lem:genparallelplay}. Of course, for $i \neq j$, there cannot be a board in $\cA_i$ and a board in $\cB_j$ forming a matching pair. Furthermore, by Theorem \ref{thm:one-vs-all-n}, each instance $(\cA_i, \cB_i)$ is winnable by $\bS$ following a strategy that has the same pattern of length $3\lceil\log_3(n)\rceil$. By Lemma \ref{lem:genparallelplay}, therefore, the total number of rounds required by $\bS$ to win the remainder of the game after the preprocessing moves have been played is $3\lceil\log_3(n)\rceil$. In total, therefore, the number of rounds required is $\lceil n/\log_r(n)\rceil + 3\lceil\log_3(n)\rceil$. The claim follows by setting $N_r = \max(N'_r, r^2)$.
%
\end{proof}
To now prove the theorem, we start by fixing $\varepsilon > 0$. Now, take $r > 2$ to be small enough so that $\log(r) < 1 + \varepsilon/2$. By Claim \ref{claim:realr}, there is some $N_r$ such that for all $n \geq N_r$, the number of quantifiers needed to separate any instance $(\cA, \cB)$ is at most:
\begin{align*}
    \left\lceil\frac{n}{\log_r(n)}\right\rceil + 3\lceil\log_3(n)\rceil &\leq \frac{n}{\log(n)}\cdot\log(r) + 3\lceil\log_3(n)\rceil + 1 \\
    &\leq \frac{n}{\log(n)}\left(1 + \frac{\varepsilon}{2}\right) + 3\lceil\log_3(n)\rceil + 1.
\end{align*}
Of course, $n/\log(n)$ asymptotically dominates $3\lceil\log_3(n)\rceil + 1$. So there is some $N'_\varepsilon$ such that for all $n \geq N'_\varepsilon$, we have $3\lceil\log_3(n)\rceil + 1 \leq (\varepsilon/2)\cdot(n/\log(n))$. So now, if $N_\varepsilon := \max(N'_\varepsilon, N_r)$, we have, for all $n \geq N_\varepsilon$:
\begin{equation*}
    \left\lceil\frac{n}{\log_r(n)}\right\rceil + 3\lceil\log_3(n)\rceil \leq \frac{n}{\log(n)}\left(1 + \frac{\varepsilon}{2}\right) + \frac{n}{\log(n)}\cdot\frac{\varepsilon}{2} = \frac{n}{\log(n)}\cdot(1 + \varepsilon).
\end{equation*}
It follows that, for any $n$, the number of quantifiers required to separate any instance $(\cA, \cB)$ of $n$-bit strings is at most $(1 + \varepsilon)\frac{n}{\log(n)} + N_\varepsilon$: the first term handles all $n \geq N_\varepsilon$, and the second term (trivially) handles all $n < N_\varepsilon$. The conclusion of the theorem follows.
\end{proof}

Remarkably, we cannot improve the upper bound in Theorem \ref{thm:all-vs-all upper bound} by any significant amount. The following proposition establishes this by means of a counting argument.

\begin{restatable}[Any vs.~Any, $n$-bit --- \emph{Lower Bound}]{proposition}{anyvsanylower}\label{prop:all-vs-all lower bound}
For all sufficiently large $n$, there is a nonempty set of $n$-bit strings, $\cA \subsetneq \{0, 1\}^n$, such that every sentence sentence $\varphi$ for $(\cA, \{0, 1\}^n - \cA)$ must have at least $n/\log(n)$ quantifiers.
\end{restatable}

\begin{proof}
Take $n$ to be sufficiently large, and suppose $k$ (as a function of $n$) is the minimum number of quantifiers that is sufficient to separate every pair of disjoint sets of $n$-bit strings. We already know $k = o(n)$ from Theorem \ref{thm:all-vs-all upper bound}, and also $k \geq \log(n)$ from Proposition \ref{prop:one-vs-one-n}. Note that this means:
\begin{equation}\label{eqHelpful}
    \log(n) > \log(k) + 2/k.
\end{equation}
We wish to show that $k \geq n/\log(n)$. To this end, consider the number of pairwise inequivalent sentences that can be written with $k$ or fewer quantifiers. Assume any such sentence is written in prenex form. It must start with a quantifier signature of length at most $k$, followed by its quantifier-free part, which can be written as a disjunction of types. The number of such quantifier signatures is $\sum_{i = 0}^k2^i < 2^{k+1}$. Any type with $k$ or fewer variables can be completely specified by fixing the relative ordering of those variables (requiring at most $k$ occurrences of the variables, using transitivity of the $\leq$ relation), and fixing each of them to be $0$ or $1$ using the appropriate unary predicate (requiring another at most $k$ occurrences). Therefore, the total number of such types is at most $k!\cdot 2^k$. Since $k! \leq (k/2)^k$ for $k \geq 6$, the total number of types is bounded above by $(k/2)^k\cdot 2^k = 2^{k\log(k)}$. Any subset of types can be in the disjunction, leading to the number of quantifier-free parts being at most $2^{2^{k\log(k)}}$. This puts the total number of pairwise inequivalent formulas using $k$ quantifiers to be at most $2^k\cdot 2^{2^{k\log(k)}}$.
    
Now, consider an instance $(\cA, \{0, 1\}^n - \cA)$, where $\cA$ is a nonempty strict subset of the $n$-bit strings. Observe that any two distinct such instances \emph{must} require inequivalent sentences to separate them. Therefore, the number of pairwise inequivalent sentences we require in order to be assured of solving the problem is at least the number of such instances, which is $2^{2^n} - 2$, where we subtract $2$ to ensure there is at least one string on either side of each such instance. It follows that we need $2^k\cdot 2^{2^{k\log(k)}}  \geq 2^{2^n} - 2 \geq 2^{2^n - 1}$, i.e.:
\begin{equation}\label{k-need}
    k + 2^{k\log(k)} \geq 2^n - 1.
\end{equation}
But if $k < n/\log(n)$, we must have:
\begin{equation*}
    k + 2^{k\log(k)} < 2^{k\log(k) + 1} < 2^{k(\log(n) - 2/k) + 1} = 2^{k\log(n) - 1} < 2^{n - 1} < 2^n - 1,
\end{equation*}
where the first inequality follows because $2^{k\log(k)} > k$, the second follows from Eq.~\eqref{eqHelpful}, the third follows by the assumption that $k < n/\log(n)$, and the fourth follows for all sufficiently large $n$. Since this contradicts Eq.~\eqref{k-need}, it follows that $k \geq n/\log(n)$, as desired. In fact, the same argument also shows that with high probability, a \emph{random} instance $(\cA, \{0, 1\}^n - \cA)$ requires at least $n/\log(n)$ quantifiers to separate.
\end{proof}
\section{Conclusions \& Open Problems }\label{sec:conclusion}

We have analyzed MS games using parallel play to obtain non-trivial quantifier upper bounds, and in two cases, lower bounds, for a variety of linear order and string separation problems. Natural directions to extend this work include the following.

\begin{itemize}
\item The gaps between the upper and lower bounds on quantifier number for string separation problems can be tightened. In the \textbf{one vs.~all} problem (Theorem \ref{thm:one-vs-all-n}), is $3\log_3(n)$ optimal? Can the upper bound (Theorem \ref{thm:all-vs-all upper bound}) for separating any two sets of strings be brought down to $n/\log(n)$, by removing the $(1 + \varepsilon)$ factor?
    
\item Learning more about how many quantifiers are needed to express particular string and graph properties would be illuminating. While our lower bound for the \textbf{one vs.~one} problem (Proposition \ref{prop:one-vs-one-n}) gave a pair of strings requiring $\log(n)$ quantifiers to separate, the counting argument used to prove Proposition \ref{prop:all-vs-all lower bound} does not exhibit a \textit{particular} instance on $n$-bit strings that requires $n/\log(n)$ quantifiers to separate.

\item It is known for ordered structures that with $O(\log n)$ quantifiers, we can express the $\bit$ predicate, or equivalently all standard arithmetic operations on elements of the universe  \cite{IMMERMAN:1999}. In particular, with $\bit$, some properties that would otherwise require $\log n$ quantifiers can be expressed using $O(\log(n)/\log\log(n))$ quantifiers. Understanding the use of $\bit$ or other numeric relations would be valuable.
\end{itemize}

\section*{Acknowledgments}

Rik Sengupta was supported by NSF
CCF-1934846. Ryan Williams was supported by NSF CCF-2127597 and a Frank Quick Faculty Research Innovation Fellowship. The authors would also like to thank Sebastian Pfau for a helpful observation, which improved the statement of the Parallel Play Lemma.

\bibliographystyle{abbrv}
\bibliography{arxiv-bibiliography}



\end{document}